\documentclass[]{article}

\usepackage{amsmath}
\usepackage{amsfonts}
\usepackage{algorithm}
\usepackage{graphicx}
\usepackage{adjustbox}
\usepackage{tikz}
\usepackage{bm}
\usepackage{booktabs}
\usepackage{multirow}
\usepackage{xcolor}
\usepackage{lmodern}
\usepackage{tikz}
\usetikzlibrary{decorations.pathmorphing}

\newcommand{\between}{\hspace{-0.5mm}\boldsymbol{\cdot}\hspace{-1.5mm}-\hspace{-1.5mm}\boldsymbol{\cdot}\hspace{-0.5mm}}
\newcommand{\curly}{\mathrel{\leadsto}}
\newcommand{\cross}[4]{(#1\rightarrow #2 \curly #3\rightarrow #4)}

\usepackage{numprint}

\usepackage{pgfplots} 

\newcommand{\opt}[1]{{#1-OPT}}

\newtheorem{lemma}{Lemma}
\newtheorem{theorem}{Theorem}

\newtheorem{corollary}{Corollary}

\newcommand{\Qed}{\hfill\rule{1.2ex}{1.2ex}}
\newcommand{\myproof}[1]{{\sc Proof:}#1 \Qed }

\newcommand{\mycite}[1]{(\cite{#1})}

\newcommand{\myquote}[1]{``#1''}
\newcommand{\comment}[1]{}
\newcommand{\FOR}{{\bf for\ }}
\newcommand{\DO}{{\bf do\ }}

\newcommand{\IF}{{\bf if\ }}
\newcommand{\THEN}{{\bf then\ }}

\newcommand{\WHILE}{{\bf while\ }}
\newcommand{\ENDWHILE}{{\bf endwhile\ }}
\newcommand{\RETURN}{{\bf return\ }}

\newcommand{\D}{\Delta}

\newcommand{\resetorigin}[2]{
	
	\tikzset{shift={(#1,#2)}}
}

\newcommand{\good}{large}

\newcommand{\prob}[1]{{\bf Pr}\left[\, #1\, \right]}
\newcommand{\probd}[2]{{\bf Pr}[\, #2\, ]_{#1}}
\newcommand{\expected}[1]{\mathbb{E}\left[\, #1\, \right]}
\newcommand{\expectedd}[2]{\mathbb{E}\left[\, #2\, \right]_{#1}}

\title{Algorithmic strategies for finding the best TSP 2-OPT move in average
sub-quadratic time}
\author{Giuseppe Lancia \thanks{Dipartimento di Matematica, Informatica e Fisica, University of Udine, 33100 Udine, Italy, {\tt giuseppe.lancia@uniud.it}}\and 
Paolo Vidoni\thanks{Dipartimento di Scienze Economiche e Statistiche, University of Udine, 33100 Udine, Italy, {\tt paolo.vidoni@uniud.it}}}
\date{}
\begin{document}
	
	\maketitle

\begin{abstract}
	We describe an exact algorithm for finding the
	best 2-OPT move which, experimentally, was observed to 
	be much faster than the standard quadratic approach. 
	To analyze its average-case complexity, we introduce a family of
	heuristic  procedures and discuss their complexity when
	applied to a random tour in graphs whose edge costs 
	are either uniform random numbers in $[0,1]$ or  Euclidean distances 
	between random points in the plane. We prove that, 
	for any probability $p$: (i) there is a heuristic in the family
	which can find the best move with probability 
	at least $p$ in average-time $O(n^{3/2}$) for uniform instances and  $O(n)$ for Euclidean instances; (ii) the exact algorithm take lesser time 
then the above heuristic  on all instances on which the heuristic finds
	the best move.
During local search, while the tour becomes less and less random, the speed of our algorithm worsens until it becomes quadratic. We then discuss how to fine tune a successful hybrid approach, made of our algorithm in the beginning followed by the usual quadratic enumeration.
\end{abstract}


\section{Introduction}

It is safe to say that there are very few acronyms in the Operations Research community 
more famous than TSP and \opt{2}. These acronyms indeed identify two among the most important subjects of our discipline, i.e., {\em combinatorial optimization} and {\em local search}.
The TSP (Traveling Salesman Problem)  is probably the 
most well-known combinatorial optimization problem \mycite{TSPbook3,TSPbook,TSPbook2}.
Its objective is to find a  shortest Hamiltonian cycle in a complete graph  of $n$ nodes  weighted on the arcs, and its importance stems from 
countless applications to all sorts of areas, like, e.g., scheduling, sequencing, routing, circuit printing, and so on up to computational biology and x-ray crystallography \mycite{TSPmatai}. Due to its relevance,
the problem has been extensively studied over the years, and
several programs  have been designed for both its exact and heuristic solution  \mycite{Lodi2007}. The most sophisticated algorithms (especially those based on mathematical programming formulations) 
have proved to be very effective, but at the price of a certain complexity,
pertaining both to their logic and their implementation.

Much in the same way as TSP is emblematic of all combinatorial optimization problems, \opt{2} is a prominent example of the concepts of neighborhood and 
local search procedure for an NP-hard problem.
Local search \mycite{AaLe97,PapSte} is a general paradigm for the  minimization of an objective function $f$ over a set ${\cal S}$ of feasible solutions. The main ingredient characterizing a local search procedure is
a map which associates to every solution $x\in {\cal S}$ a set $N(x)\subset {\cal S}$ called its {\em neighborhood}. Starting at a solution $x^0$, local search samples the solutions in $N(x^0)$ looking for a solution $x^1$ better than $x^0$. If it finds one such solution, it iterates the same step, this time looking for $x^2$ in $N(x^1)$, and then continues the same way until the current solution $x^i$ satisfies $f(x^i)=\min\{f(x) | x\in N(x^i)\}$, i.e., it is a {\em local optimum}. Replacing $x^i$ with $x^{i+1}$ is called performing a {\em move} of the search, and the total number of moves applied to get from $x^0$ to a local optimum is called the {\em  convergence length}.
For a small-sized neighborhood, local search usually follows a  {\em best-improvement}  strategy, i.e., $x^{i+1}$ is the  the best 
solution possible in $N(x^i)$.

The \opt{2} neighborhood associates to 
each TSP solution (also called a {\em tour}) the set of all tours that 
can be obtained by removing two edges and replacing them with two new ones.
Let us assume the set of vertices to be $[n]:=\{1,\ldots,n\}$. Then a tour is defined by a permutation     $T=(\pi_1,\ldots,\pi_n)$ of $[n]$ and a 2-OPT move $\mu(i,j)$ is 
identified by two non-consecutive edges of the tour, namely 
$\{\pi_i,\pi_{i+1}\}$ and $\{\pi_j,\pi_{j+1}\}$, called the {\em pivots} of the move.
The move removes $\{\pi_i,\pi_{i+1}\}$ and $\{\pi_j,\pi_{j+1}\}$
and replaces them with $\{\pi_i,\pi_{j}\}$ and 
$\{\pi_{i+1},\pi_{j+1}\}$, yielding the new tour $T'$ (see the figure below).

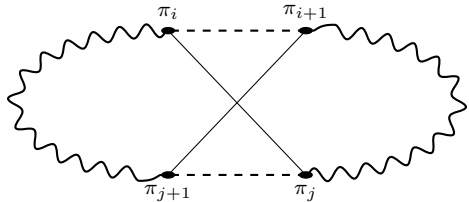
\tikzset{snake it/.style={decorate, decoration=snake}}

\begin{center}
	{\footnotesize
		\begin{tikzpicture}[xscale=1.22,yscale=0.6]			
			\draw[ thick, snake it] (0,0) arc (-90:90:1.6);
			\node at (0,-0.4) {$\pi_j$};
			\node at (-1.5,-0.4) {$\pi_{j+1}$};
			\draw[ thick, snake it] (-1.5,3.2) arc (90:270:1.6);
			\draw [dashed,thick] (0,0) -- (-1.5,0);
			\node at (0,3.6) {$\pi_{i+1}$};
			\node at (-1.5,3.6) {$\pi_i$};
			\draw [dashed,thick] (0,3.2) -- (-1.5,3.2);
			\filldraw (0,0) circle (2pt);
			\filldraw (-1.5,0) circle (2pt);
			\filldraw (0,3.2) circle (2pt);
			\filldraw (-1.5,3.2) circle (2pt);
			\draw  (0,0) -- (-1.5,3.2);
			\draw  (-1.5,0) -- (0,3.2);
		\end{tikzpicture}
	}
\end{center}

In this paper we  consider  the symmetric TSP, i.e., the graph is undirected, so that the distance between any two nodes is the same in both directions.
Denote by $c(i,j)=c(j,i)$ the distance between two generic nodes $i$ and $j$.  
The length $c(T)$ of a tour $T=(\pi_1,\ldots,\pi_n)$ is the sum of the lengths of the tour edges, i.e., the edges $\{\pi_i, \pi_{i+1}\}$, for $i=1,\ldots,n$ (where we assume $\pi_{n+1}:=\pi_1$). 
For a \opt{2} move $\mu(i,j)$, we define
\[
\Delta(\mu(i,j)) := c(T) - c(T') = c(\pi_i,\pi_{i+1}) + c(\pi_j,\pi_{j+1}) - 
\big(c(\pi_i,\pi_{j}) + c(\pi_{i+1},\pi_{j+1})\big).
\]
We say that the move is {\em improving} if $\Delta(\mu(i,j)) > 0$. An improving move is {\em best improving} if $\Delta(\mu(i,j))=\max_{u,v} \Delta(\mu(u,v))$, and the goal is to find a best-improving move.

The introduction of the \opt{2} neighborhood for the TSP 
dates back to the late fifties \mycite{flood,croes58}, and still today local search based on this neighborhood is probably the most popular approach for the TSP (especially on large instances),  for reasons of simplicity, low time-complexity and overall effectiveness. 
Indeed, there are some more sophisticated heuristics for
the TSP, such as, e.g., \opt{3} \mycite{Lin65}, or metaheuristics like genetic algorithms \mycite{TSPgenetic,TSPgenetic1}, simulated annealing \mycite{TSPSA} and tabu search \mycite{TSPtabu}. The most effective
heuristic procedure, i.e., the one for which the trade-off betweeen  quality of  solutions found and time spent in finding them is 
the best, is  a  \opt{3} variant known as Lin-Kernighan's algorithm
\mycite{LinKer,LinKer1}. However, all these sophisticated heuristics are somewhat 
complex to understand and implement, especially in comparison with the simplicity of \opt{2}. This aspect is considered very important in a large part of the industrial world, where in-house software development and  maintenance
oftentimes lead to the adoption of simple, yet effective, solutions like 
some basic local search. Indeed, in the case of \opt{2}, the algorithm
to find the best move $\mu(i,j)$ is trivial
i.e., a nested-{\tt for} cycle iterating over all $1\le i < j \le n$ and taking $\Theta(n^2)$ time.  Since, generally speaking, quadratic algorithms are  considered very effective, not very much research went into trying to speed-up the algorithm for finding the best move.
In this work, however, we will prove that speeding-up the standard quadratic
procedure for \opt{2} is {\em extremely} simple, at least as far as 
implementation goes, since it amounts to a trivial \myquote{hack} of the
nested-{\tt for} procedure consisting in the addition of just one {\tt if()} statement.
Besides this simple new version of the search procedure, we will also describe a slightly more sophisticated one, which achieves an even better 
performance. A probabilistic analysis shows that our strategies do in fact change the order of complexity from quadratic to sub-quadratic on average. In particular, 
for Euclidean TSP instances we manage to find the best move on a random tour in average time $O(n)$, i.e., the best possible complexity given that the tour has $n$ edges that must be looked at.

Most of the literature on \opt{2} focuses on the study of the
convergence length and the quality of the 
local optima that can be obtained.
In particular, \cite{chandraetal} (extending a result of \cite{kern}) have shown that the length of convergence is 
polynomial  on average for random Euclidean instances, while \cite{englertEtAl}  have shown how to build very
particular Euclidean instances on which the length of the 
convergence is exponential.
As for the local optima quality, \cite{chandraetal} shows that they are,  with high probability,  within a constant factor from the global optimum for random Euclidean instances, while worst-case factors depending on $n$ are given in \cite{thesis17}.

With regard to the time spent in finding the best move at each local 
search iteration, the nested-{\tt for} algorithm is not only worst-case $\Theta(n^2)$, but its  average-case is $\Theta(n^2)$ as well. Building on our previous research 
\mycite{LanciaVidoni_2020} 
in which we studied how to speed-up some enumerative algorithms
looking for a best solution in a polynomially-sized search space, 
we propose two variants of a new  algorithm (one guided by a greedy criterion while the other ``blind''), for
finding the best 2-OPT move. We give empirical and theoretical
evidence that our algorithm is better than quadratic,  on average, when looking for the best move on a long sequence of the tours visited by local search (roughly, two thirds of the convergence length)  starting at a random tour. 
In particular, at the initial stages of the local search, 
it takes us about $O(n^\frac{3}{2})$ to find the best move  for 
graphs  whose distances are drawn uniformly at random (u.a.r.) in $[0,1]$, and about $O(n)$ for
Euclidean graphs whose vertices are points 
drawn u.a.r. in  the unit square. In order to 
perform an average-case analysis, we introduce a family of heuristics
for finding the best \opt{2} move, and discuss both their average-case running time and their probability of success, describing how we can control
both these aspects.

On medium- to large-sized instances, our procedure can achieve speedups
of { two to three orders of magnitude} 
over the nested-{\tt for} algorithm for 
most of the convergence. However, while the search progresses and we near the local optimum, our algorithm becomes less effective, so that at some point it might be better to switch back to the $\Theta(n^2)$ enumerative procedure since it does not have the
overhead of dealing with our data structure (namely, a heap). Experimentally, we determined that this phenomenon happens in the final part of the convergence. It remains then an interesting research question the design
of an effective algorithm to find the best \opt{2} move for a nearly locally optimum tour.

\paragraph{Paper organization.} The remainder of the paper is organized as follows. In Section \ref{sec:idea} we describe the general
idea for searching the best 2-OPT move, and present the algorithm, in its two versions. Section \ref{sec:prob} is devoted to the probabilistic analysis of our algorithm and of a family of heuristics for the problem. In Section \ref{sec:results} we report on our computational experiments and the
statistical results that we obtained. Some conclusions are
drawn in Section \ref{sec:concl}.

\section{Our strategy for moves enumeration}
\label{sec:idea}

Without loss of generality, let us assume that the tour is $T=(1,\ldots,n)$.
For $\alpha>0$ let us call {\em $\alpha$-\good} any edge $\{i,i+1\}$ along the tour, such that $c(i,i+1) > \alpha$. 
In this paper we are going to follow a strategy that allows us not to enumerate all moves, but only those which are ``good candidates'' to be the
best overall. The idea is quite simple, and it relies on a sequence of iterative improvements in which, at each iteration, there is a certain move (the current ``champion'') which is the best we have seen so far and which  we want to beat.

Assume the current champion is $\hat \mu:=\mu(\hat \imath , \hat \jmath)$.
Then, for any move $\mu(i,j)$ better than $\hat \mu$ it must be
\begin{align*}
	\Delta(\hat \mu)  & < \Delta(\mu(i,j)) \\
	& =  c(i,i+1) + c(j,j+1) - 
	\big(c(i,j) + c(i+1,j+1)\big) \\
	& \le  c(i,i+1) + c(j,j+1) 
\end{align*}
and hence
\[
\left(c(i,i+1) > \frac{\D(\hat \mu)}{2}\right) \lor \left(c(j,j+1) > \frac{\D(\hat \mu)}{2}\right)
\]
i.e., at least one of the move pivots must be  $(\D(\hat \mu)/2)$-\good.  Based on this observation,  we will set-up an enumeration scheme which builds the moves starting from pivots that are  $(\D(\hat \mu)/2)$-\good\ and then completing any such edge into a move by adding the second pivot of the move.
For simplicity, from now on we will be often referring to  
$(\D(\hat \mu)/2)$-\good\ edges
simply  as 
{\em \good\ edges}.

Our basic steps are the  {\em selection} and the {\em expansion} of the tour edges. 
The selection of an edge is simply the choice of an edge $\{i,i+1\}$ (which has not been selected before). The expansion of $\{i,i+1\}$ is the evaluation of all moves  $\mu(i,j)$, for $j\ne i$. Clearly an expansion costs $\Theta(n)$ while the selection cost depends on the criterion we adopt.
In particular, we propose two versions of our 
algorithm, namely the {\em greedy} and the {\em blind} versions.
In our greedy  algorithm the cost for a selection will be $O(\log n)$, while in the blind algorithm it will be $O(1)$.

Each algorithm, at a high level, can be seen as a sequence of iterations, where each iteration is a selection followed, perhaps,  by an expansion. 
We will discuss how, given a random starting tour, as long as we expand only \good\ pivots we expect to have, overall, only
$O(n^{1-\epsilon})$, with $\epsilon>0$, expansions, even if we pick the pivots without following any particular criterion, i.e., blindly. 
If, on the other hand, we use the greedy criterion of selecting the pivots from the largest to the smallest, we still expand only $O(n^{1-\epsilon})$ pivots, but this time the multiplicative constant is better, i.e., we get a faster algorithm. 

\floatname{algorithm}{Procedure}
\begin{algorithm}[t]
	\caption{\label{alg:find3}
		{\sc GreedyHeapBasedAlgorithm} ${\cal A}_g$}
	\begin{quote}
		\noindent\\
		\mbox{\ }1.\hspace{2mm}
		Build a max-heap with  elements $[i,c(i,i+1)]$ $\forall 1\le i \le n$\\ 
		\mbox{\ }\phantom{1.}\hspace{3mm}
		sorted by $c$-values;\\ 
		\mbox{\ }2.\hspace{2mm}
		Set $\hat{\mu}:=\emptyset$ and $\D(\hat{\mu}):=-\infty$; {\small \tt /* first undefined champion */} \\ 
		\mbox{\ }3.\hspace{2mm}
		\WHILE the $c$-value of the top heap element is $>\D(\hat{\mu})/2$ \DO\\
		\mbox{\ }4.\hspace{8mm}
		Extract  the top of heap, let it be
		$[i,c(i,i+1)]$; \\
		\mbox{\ }5.\hspace{8mm}
		\FOR $j:=1,\ldots i-2, i+2 \ldots n$ \DO\\
		\mbox{\ }6.\hspace{14mm}
		\IF $\D(\mu(i,j)) > \D(\hat{\mu})$ \THEN \\ 
		\mbox{\ }7.\hspace{20mm}
		$\hat{\mu}:=\mu(i,j)$; {\small \tt /* update the champion */}\\
		\mbox {\ }8.\hspace{2mm}
		\ENDWHILE\\
		\mbox {\ }9.\hspace{2mm}
		\RETURN  $\hat{\mu}$; 
	\end{quote}
	\label{alg:greedy}
\end{algorithm}

\subsection{The greedy algorithm}

The greedy algorithm ${\cal A}_g$ is described in Procedure \ref{alg:greedy}.
In this algorithm we make use of a max-heap, in which we put each node $i\in\{1,\ldots,n\}$ together with the value $c(i,i+1)$, used to order the heap.
Each entry of the heap is indeed a 2-field record $[i,c(i,i+1)]$ (where, by convention, we define $n+1$ to be $1$). 
Building the heap (step 1) can be done in linear time with respect to the number of its elements (i.e., in time $\Theta(n)$ in our case) by using  standard procedures \mycite{cormen}.

At the beginning $\hat\mu$ is undefined, so we set $\D(\hat\mu):=-\infty$.
Testing if there are still any \good\ edges is done in step 3 and takes time $O(1)$ per test since we just need to read the value at the heap's root. The selection is done in step 4 and takes time $O(\log n)$ to maintain the heap property.
The main  loop 3--8 terminates as soon as there are no longer any \good\  edges.

At the generic step, we pop the top of the heap, let it be $[i,c(i,i+1)]$. If 
$c(i,i+1)\le \D(\hat\mu)/2$, we stop and return $\hat\mu$ as the best possible move. Otherwise, $\{i,i+1\}$ is one of the  two pivots of some  move potentially better than $\hat\mu$.
Knowing one pivot,  we then run the expansion (steps 5--7) which, in linear time, finds the best completion of  $\{i,i+1\}$ into a move $\mu(i,j)$.
Each time we find a move better than  the current champion, we update $\hat\mu$. This way the termination condition 
becomes easier to satisfy and we get closer to the end of the loop.

If, overall, there are $N_g$ selections (and, therefore $N_g$ expansions), the running time of the algorithm is $O(n + N_g(\log n + n))$ which is better than quadratic as long as $N_g=O(n^{1-\epsilon})$ for some $\epsilon > 0$.

\subsection{The blind algorithm}

The blind algorithm, called ${\cal A}_b$, is outlined in Procedure \ref{alg:blind}. The main loop consists of ${n}$ iterations. At each iteration the selection of an edge $\{i,i+1\}$ has cost $O(1)$. The expansion is done in steps 4--6. Assuming there are $N_b$ expansions altogether, the running time of this algorithm is
$O(n + N_b n)$, which is better than quadratic as long as $N_b =O(n^{1-\epsilon})$.
Notice that if we remove line 3, and start the {\tt for}-cycle of line 4 at $j:=i+2$ instead than $j:=1$, we obtain exactly the standard two-{\tt for}
quadratic procedure for 2-OPT optimization. It is then evident how simple it is to implement this procedure,
since it basically requires just to add an {\tt if} to
the standard algorithm.

\floatname{algorithm}{Procedure}
\begin{algorithm}[t]
	\caption{\label{alg:find2}
		{\sc BlindAlgorithm} ${\cal A}_b$}
	\begin{quote}
		\noindent\\
		\mbox{\ }1.\hspace{2mm}
		Set $\hat{\mu}:=\emptyset$ and $\D(\hat{\mu}):=-\infty$; {\small \tt /* first undefined champion */}\\ 
		\mbox{\ }2.\hspace{2mm}
		\FOR $i:=1,\ldots,n$ \DO\\
		\mbox{\ }3.\hspace{8mm}
		\IF $c(i,i+1) > \D(\hat{\mu})/2$ \THEN\\
		\mbox{\ }4.\hspace{14mm}
		\FOR $j:=1,\ldots,i-2,i+2,\ldots,n$ \DO\\
		\mbox{\ }5.\hspace{20mm}
		\IF $\D(\mu(i,j)) > \D(\hat{\mu})$ \THEN \\ 
		\mbox{\ }6.\hspace{26mm}
		$\hat{\mu}:=\mu(i,j)$; {\small \tt /* update the champion */}\\
		\mbox{\ }7.\hspace{2mm}
		\RETURN  $\hat{\mu}$; 
	\end{quote}
	\label{alg:blind}
\end{algorithm}

\section{Probabilistic analysis}
\label{sec:prob}

\subsection{The general plan}

In this section we discuss the average-case  complexity of our greedy algorithm, obtaining some theoretical justification of the
empirical evidence, i.e., that it is
better than quadratic for a large portion of the 
convergence to a local optimum. In particular, at the very first steps, when the current solution is a random (or almost random) tour, we observed an average complexity of $O(n^{3/2})$ on
uniform instances and $O(n)$ on Euclidean instances.
(The latter is indeed an optimal result for  this problem, since we should not expect to be able to find the best move in a shorter time  than that required to look at all the edges of the tour). Note that the blind algorithm ${\cal A}_b$ was observed to
have the same behavior, but with a slightly worse 
multiplicative constant, so we can focus on 
${\cal A}_g$. 

The following analysis is relative to the problem of finding the best move on a { random} tour. In order to explain the observed sub-quadratic complexity, we start by 
discussing weaker versions of the algorithm and prove that 
they run in average sub-quadratic time. These algorithms are 
heuristics, but we can make their probability of success as high
as we please. We then show that the running time of any of these heuristics upper bounds the running time of the greedy algorithm  on those instances on  which the heuristic succeeds (which, as we already remarked, can be almost all).

\newcommand{\ALG}{\textrm{ALG}}

\paragraph{Some preliminaries.} In the average-case analysis of an algorithm one considers instances drawn at random according to a certain probability distribution. In our study an instance is given by 
${n\choose 2}$ non-negative reals (representing the edge lengths/costs of a complete undirected graph of $n$ nodes) plus a permutation of $\{1,\ldots,n\}$ identifying a tour in the graph. 
The size of an instance can be characterized through a parameter $n$ which, for us, is the number of nodes in the graph. By a {\em random tour} we denote a permutation drawn u.a.r. in the set including all the permutations. Hereafter, we assume $n\in \mathbb{N}$, with $n\ge 4$, and, when we talk of a generic instance and of $n$ in the same sentence, $n$ is the instance size. As far as the edge costs are concerned, we consider two types of distributions:
\begin{enumerate}
	\item {\bf Uniform instances:} A random instance of this type is obtained by setting the cost of each edge $\{i,j\}$ to a value drawn u.a.r. in $[0,1]$. 
	Note that the edge lengths are independent random variables.
	\item {\bf Euclidean instances:} A random instance of this type is obtained by drawing u.a.r. $n$ points $P_1,\ldots,P_n$ in the unit square and then setting the cost of each edge $\{i,j\}$ to the Euclidean distance between $P_i$ and $P_j$. 
	Note that the edge lengths are not independent random variables since  triangle inequality must hold.
\end{enumerate}

In the following analysis, we denote by $t^n_{\cal A}(I)$ the time (i.e., number of elementary steps) taken by the algorithm $\cal A$ on an instance $I$ of size $n$ and we define the associated random variable $T^n_{\cal A}$ as the time taken by $\cal A$ on a random instance of size $n$. The random instance is generated according to a suitable probability distribution corresponding to the uniform or to the Euclidean instances framework. The average-case complexity of the algorithm $\cal A$ is then defined as  
\[
\bar T_{\cal A}(n):=\expectedd{}{T_{\cal A}^n},
\]
interpreted as a function of the size $n$. We count as the elementary steps yielding the complexity of an algorithm the number of moves that it evaluates (which, 
in turn, is a factor-$n$ proportional to the number of edges expanded).
Notice that, although an expansion for ${\cal A}_g$ requires also 
a work of cost $O(\log n)$ to determine which edge is expanded and to rearrange the heap, this $O(\log n)$ factor is dominated by the $\Theta(n)$ work due to the moves evaluated by the expansion.

\paragraph{A family of \myquote{fixed threshold} heuristics.} Let us consider a variant of our algorithm which works as follows:
Given an input $I$ of size $n$, the algorithm 
first computes a threshold $\delta_n$ (i.e., depending only on $n$, and constant for a fixed $n$) and then it expands all and only the edges $\{i,j\}$ of the tour such that $c(i,j)>\delta_n$. Notice that there is an algorithm of this type for each possible function $\delta_n$, and hence we can talk of a family of algorithms. Let us call a generic  algorithm of this family \ALG($\delta_n$). We remark that there is no heap, and the algorithm is a simple loop like  the blind version, but does not use the updated champions to set new thresholds to beat. 

Each algorithm \ALG($\delta_n$) can be seen as a heuristic for finding the best \opt{2} move. Indeed, there is no guaranteed 
that it will find the best move, but rather it will find it with a certain probability, depending on $\delta_n$ and on the distribution of instances. In particular, \ALG($\delta_n$) may fail to find the best move because of one of two types of errors, i.e., 
\begin{itemize}
	\item [ERR$_0$:] when no edge is expanded (all arcs have cost $\le \delta_n$) and hence no move will be found. 
	\item [ERR$_1$:] when  some edges are expanded, but the optimal move did not remove any edges of length $> \delta_n$ and so it won't be found. 
\end{itemize}

The probability of  failure can be controlled by a proper setting of $\delta_n$. Intuitively, by lowering (increasing) $\delta_n$ we decrease (respectively, increase) the probability of errors. At the same time, we  increase (respectively, decrease) the average time complexity of the algorithm, since more (respectively, less) edges get expanded. We will describe a way to balance these two conflicting objectives, namely, having a $\delta_n$ large enough so as to guarantee 
an average sub-quadratic algorithm, but small enough so as the
probability of errors can be upper-bounded by any given constant.

\ \\

Let $\Delta^*(I)$ denote the value of an optimal  2-OPT move on an instance $I$.
The following is a sufficient, but not necessary, condition for \ALG($\delta_n$) to find the optimal solution:

\begin{lemma}
	For every instance $I$ for which $\delta_n < \Delta^*(I)/2$, 	
	\ALG($\delta_n$) finds
	an optimal solution.
	\label{lem:gooddelta}
\end{lemma}
\myproof{	Assume a best move is $\mu^*(i,j)$. Then either 
	$c(i,i+1)\ge \Delta^*(I)/2$ or $c(j,j+1)\ge \Delta^*(I)/2$. Therefore one of the two edges will
	be expanded by \ALG($\delta_n$) and $\mu^*$ will
	be found.}

Given an instance $I$, let us call {\em good move} any move $\mu$ such that $\Delta(\mu) > 2\delta_n$ (notice that the property of being good for a move depends on $\delta_n$, but, for simplicity, we assume that $\delta_n$ is implicit from the context). Let us also call {\em good instance} any instance for which there exists at least one good move.
Then we have the following
\begin{corollary}
	\label{cor:goodmove}
	For every good instance $I$, 	
	\ALG($\delta_n$)  finds	an optimal solution.
\end{corollary}
\myproof{
	Let $\mu$ be a good move in $I$. Then, $2\delta_n < \Delta(\mu) \le \Delta^*(I)$ and the conclusion follows from Lemma \ref{lem:gooddelta}.}

Furthermore, under the conditions of Corollary \ref{cor:goodmove}, we 
are sure that ${\cal A}_g$ runs faster than \ALG($\delta_n$).

\begin{lemma}
	\label{lem:dominate}
	For every good instance $I$, it is $t^n_{{\cal A}_g}(I) \le t^n_{\ALG(\delta_n)}(I)$.
\end{lemma}

\myproof{
	Since there exist good moves, it is  $\delta_n \le \Delta^*(I)/2$.
	Let $e_1, \ldots, e_k$ be the sequence of edges expanded by ${\cal A}_g$. Note that  $c(e_1) \ge \cdots \ge c(e_k)$. Independently of which $e_p$, for $1\le p\le k$, is the edge whose expansion yields the optimal move, since $e_k$ was eventually expanded it must be $c(e_k)\ge  \Delta^*(I)/2$. Therefore 
	$c(e_i)\ge  c(e_k) \ge \Delta^*(I)/2 > \delta_n$ for all $i=1,\ldots,k$.
	Since   \ALG($\delta_n$) expands all edges of value $>\delta_n$,
	this implies that all edges expanded by 
	${\cal A}_g$ are also expanded by 
	\ALG($\delta_n$).}

The following lemma is useful for evaluating the average-case complexity of \ALG($\delta_n$) for every distribution over the instances.
\begin{lemma}
	\label{lem:time}
	Let $C$ be the random variable representing the cost of the edge between two random nodes of the graph. If $\delta_n$ is chosen so that	$\probd{}{C > \delta_n} = \Theta(n^{-r})$, with $r\in (0,1]$, then 
	$\bar T_{\ALG(\delta_n)}(n)=\Theta(n^{2-r})$.
\end{lemma}
\myproof{
	For each  $1\le i < j\le n$, consider the indicator variable $X_{ij}$ which is 1 if the edge $\{i,j\}$ is in the tour and its length is greater than $\delta_n$. Since these two events are independent, it is
	\[
	\expected{X_{ij}} = \prob{X_{ij}=1} = \frac{n\, \prob{C>\delta_n}}{{n\choose 2}}  \quad\textrm{ for each } 1\le i<j\le n.
	\]
	Let $Y=\sum_{ij} X_{ij}$ be the random variable representing how many edges  get expanded by \ALG($\delta_n$). It is
	\[
	\expected{Y} =\sum_{ij} \expected{X_{ij}} = {n\choose 2}\times
	\frac{n\,  \prob{C>\delta_n}}{{n\choose 2}} = n\,  \prob{C>\delta_n}.
	\]
	Since the expansion of a generic node $i$ involves $n-3$ nodes, namely all the nodes except $i-1$, $i$ and $i+2$, $T^n_{\ALG(\delta_n)} = (n-3) Y$. Then, we have
	\[
	\bar T_{\ALG(\delta_n)}(n) = \expected{T^n_{\ALG(\delta_n)}} = (n-3)\times \expected{Y} = \Theta(n^2)\, \prob{C>\delta_n}
	\]
	from which the conclusion follows.
}

%
\ \\

In the next two sections, we study our algorithms with respect to uniform and Euclidean random instances. For both types of instance distributions, we will  use the following approach:
\begin{enumerate}
	\item We set $\delta_n$ so that 
	$\prob{C >\delta_n}=\alpha n^{-r}$ for some
	constant $\alpha>0$ and  $r\in (0,1]$.
	\item We describe a specific type of good moves
	%
	and show that, asymptotically in $n$, the probability of having no good moves of our type tends to $0$ for increasing $\alpha$.
	This implies that for every $p\in [0,1)$ we can find an $\alpha$ to set $\delta_n$ so that, asymptotically, the probability 
	for an instance to be good
	is grater than $p$.
	\item We conclude that
	\ALG($\delta_n$) is a heuristic whose average-case running time is sub-quadratic that succeeds on at least a fraction $p$ of instances. 
	By Lemma \ref{lem:dominate}, this implies that, for at least a fraction $p$ of all instances,  ${\cal A}_g$ is dominated by an
	algorithm of sub-quadratic average-case running time, where $p$ can be made as close to 1 as we want.
\end{enumerate}

\subsection{Uniform random costs}
\label{sec:probuni}

As an immediate consequence of Lemma \ref{lem:time} we have the following
\begin{corollary} 
	\label{cor:tU}
	Let $\alpha>0$ be a constant and define
	\[
	\delta_n := 1 - \alpha\, n^{-1/2}.
	\]
	Then, under the uniform distribution setting for random instances, the average-case complexity of \ALG($\delta_n$) satisfies $\bar T_{\ALG(\delta_n)}(n)= \Theta(n^{3/2})$.
\end{corollary}
\myproof{
	Since the cost of each edge is drawn u.a.r. in $[0,1]$ it is 
	\[
	\probd{}{C > \delta_n} = 1 - \delta_n = \alpha n^{-1/2}  = \Theta(n^{-1/2}). \]
	Then, by Lemma \ref{lem:time}, it is  $\bar T_{\ALG(\delta_n)}(n)= \Theta(n^{2-1/2}) = \Theta(n^{3/2})$.
}

Let us call an edge $\{i,j\}$ {\em long} if 
\[
c_{ij} > \frac{1+\delta_n}{2} = 1 - \left(\frac{\alpha}{2}\right)n^{-1/2}
\]
and {\em short} if 
\[
c_{ij} < \frac{1-\delta_n}{2} =  \left(\frac{\alpha}{2}\right)n^{-1/2}.
\]
The specific type of good moves that we consider, called \textit{Long-Short moves} (LS-moves), are those that replace two long edges with two short ones. Indeed, for any such move $\mu$, it is
\[
\Delta(\mu) > 2(1 + \delta_n)/2 - 2( 1-\delta_n)/2 = 2\delta_n.
\]


\begin{theorem}
	For each $\alpha >0$ denote by $P_0(\alpha,n)$ the probability that 
	there is no LS-move in a random tour of $n$ nodes on a random uniform instance. Then
	\[
	\lim_{n\rightarrow \infty} P_0(\alpha,n) \le  \frac{1}{{e^{(\alpha/4)^4}}}. 
	\]
\end{theorem}
\myproof{
	Let $p_n:=(\alpha/2)n^{-1/2}$ be the probability for an edge to be long, which is the same as the probability to be short. Let the tour be $\pi = (\pi_1,\ldots,\pi_n,\pi_1)$. Then, $P_0(\alpha,n)$ is the probability that there is no cycle $(\pi_i,\pi_{i+1},\pi_{j+1},\pi_j,\pi_i)$ of four edges (two in the tour and two not in the tour) such that $\{\pi_i,\pi_{i+1}\}$ and $\{\pi_j,\pi_{j+1}\}$ are long while $\{\pi_{i+1},\pi_{j+1}\}$ and $\{\pi_i,\pi_j\}$ are short.

	We can then think of two Bernoulli trials, in sequence, where the first trials determine long edges along the tour and the second trials determine 
	good moves for pairs of long edges along the tour. To obtain independence for the second trials, we will consider moves that remove either two odd-indexed edges or two even-indexed edges. Let us focus on the odd-indexed edges.
	The first set of Bernoulli trials is repeated $n/2$ times, i.e., for all edges $\{\pi_i,\pi_{i+1}\}$ where $i$ is odd, and the probability of success is $p_n$. We have a success if the edge $\{\pi_i,\pi_{i+1}\}$ is long. Assume there have been $k$ successes altogether. Then,  the second Bernoulli trials are repeated ${k\choose 2}$ times, one for each pair $\{\pi_i,\pi_{i+1}\}$, $\{\pi_j,\pi_{j+1}\}$ of long edges . The probability of success is $p_n^2$, and there is a success if both $\{\pi_{i+1},\pi_{j+1}\}$ and $\{\pi_i,\pi_j\}$ receive a short length. Note that these trials are independent, since for every two pairs $P$ and $P'$ of long tour edges, the sets of non-tour edges that define the \opt{2} move for $P$ and for $P'$ are disjoint. 
	
	The probability of having no LS-moves at all is upper bounded by the probability of having no LS-moves of the above type, i.e., having zero successes in the second Bernoulli trials. By the laws of
	binomial distributions, this probability is 
	\[
	P'_0(\alpha,n) = \sum_{k=0}^{n/2} {\frac{n}{2}\choose k} p_n^k (1-p_n)^{\frac{n}{2}-k} (1-p_n^2)^{k\choose 2}.
	\]

	Let $S_n\sim\textrm{Binomial}(n/2,p_n)$. Then $P'_0(\alpha,n)$
	can be written as
	\[
	P'_0(\alpha,n)=\expected{\Big(1-\alpha^2/(4n)\Big)^{S_n \choose 2}}.
	\]
	
	Now fix $\delta\in(0,\frac{1}{4})$, 
	and define $a_n=\alpha\sqrt{n}/4-n^{1/4+\delta}$
	and $b_n=\alpha\sqrt{n}/4+n^{1/4+\delta}$. Also, consider the event $A_n=\{S_n\in[a_n,b_n]\}$, and let its complement be $A^c_n$. By the Chebyshev's inequality,
	
	\[
	\prob{A^c_n}\le \frac{\textbf{Var}(S_n)}{n^{\frac{1}{2}+2\delta}} \le \frac{C}{n^{2\delta}}
	\]
	for some absolute constant $C>0$
	and hence $\prob{A^c_n}\rightarrow 0$
	as $n\rightarrow \infty$. 
	Then by noting that $k\mapsto (1-p^2)^{k\choose 2}$
	is decreasing in $k$, we get
	\[
	(1-\alpha^2/(4n))^{b_n\choose 2} \prob{A_n}\le P'_0(\alpha,n) \le (1-\alpha^2/(4n))^{a_n\choose 2} \prob{A_n}+\prob{A^c_n},
	\]
	and it is easy to check that both the lower and upper bound converge to $e^{-\alpha^4/2^8}$, which is the limit of $P'_0(\alpha,n)$. Since 
	$P_0(\alpha,n)\le P'_0(\alpha,n)$, the conclusion follows.
}

\begin{corollary}
	\label{cor:pgoodU}
	For each $p\in[0,1)$ there exist a value for $\alpha>0$ and an integer $n_0$ such that, for each  $n\ge n_0$ and  uniform random instance of size $n$, it is
	$\probd{}{\textrm{The instance is good }} > p$.
\end{corollary}
\myproof{
	Let $\alpha>0$ be such that $e^{-(\alpha/4)^4}> 1-p$, i.e., $\alpha > 4 \sqrt[4]{\ln(1-p)^{-1}}$.
	Then 
	\begin{align*}
		\lim_{n\rightarrow \infty} \probd{}{\textrm{There is a least one LS-move}} &  = 1 - \lim_{n\rightarrow \infty} P_0(\alpha,n)\\
		& \geq 1 - \frac{1}{{e^{(\alpha/4)^4}}}\\
		& > p
	\end{align*}
	and therefore, from some $n_0$ on, it is $\probd{}{\textrm{There is a least one LS-move}} > p$.
	Since LS-moves are good moves, the conclusion follows.
}

Notice that we are discussing a lower bound to the probability of some specific good moves, and these are in turn a subset of  all good moves, so that it is possible to obtain the same probability of no errors with an $\alpha$ smaller than that suggested by the corollary, as we will show in our computational experiments.

Finally, the following theorem bounds probabilistically the complexity of ${\cal A}_g$ on uniform instances via $\Theta(n^{3/2})$ functions.

\begin{theorem}
	\label{teo:avgU}
	Consider the uniform distribution setting for random instances. Then, for each $p\in[0,1)$ there exists an algorithm \ALG($\delta_n$),
	with  $\bar T_{\ALG(\delta_n)}(n)=\Theta(n^{3/2})$, and an integer $n_0$ such that, for each $n\ge n_0$, it is
	\[
	\probd{}{T^n_{{\cal A}_g} \le T^n_{\ALG(\delta_n)}}  > p.
	\]
\end{theorem}
\myproof{
	By Corollary \ref{cor:pgoodU} we can find $\alpha>0$ and $n_0$ such that, for $n\ge n_0$, the probability of a good instance is greater than $p$. If we set $\delta_n:=1 -\alpha n^{-1/2}$, by Corollary \ref{cor:tU} we have $\bar T_{\ALG(\delta_n)}(n) = \Theta(n^{3/2})$. Since $T^n_{\cal A}$ is defined as the (random) time taken by an algorithm $\cal A$ on a (random) instance of size $n$, we can usefully distinguish between the case in which the instance is good and the case in which it is not good, so that 
	\begin{align*}
		\probd{}{T^n_{{\cal A}_g}\le T^n_{\ALG(\delta_n)}} & = \probd{} {\textrm{instance is good}}\times \probd{}{T^n_{{\cal A}_g} \le T^n_{\ALG(\delta_n)}\, |\, \textrm{instance is good}}\\
		& + \probd{}{\textrm{instance is not good}} \times \probd{}{T^n_{{\cal A}_g} \le T^n_{\ALG(\delta_n)}\, |\, \textrm{instance is not good}}.
	\end{align*}
	By Lemma \ref{lem:dominate}, $\probd{}{T^n_{{\cal A}_g} \le T^n_{\ALG(\delta_n)}\, |\, \textrm{instance is good}}=1$ and then, $\forall n\ge n_0$,
	\begin{align*}
		\probd{}{T^n_{{\cal A}_g}\le T^n_{\ALG(\delta_n)}} 
		& \ge \probd{}{\textrm{instance is good}} \times \probd{}{T^n_{{\cal A}_g} \le T^n_{\ALG(\delta_n)}\, |\, \textrm{instance is good}}\\
		& = \probd{}{\textrm{instance is good}} \\
		& > p.
	\end{align*}
}

\subsection{Random Euclidean instances}
\label{sec:probgeo1}

We start with a lemma whose proof is in Appendix A.
\begin{lemma}
	\label{lem:maxd1}
	Let $1.055 < d \le \sqrt{2}$ and let $D$ be the distance between two random points drawn uniformly in the unit square. Then  
	\[
	\prob{D > d} \le  \frac{7}{16} \left(1 - \sqrt{d^2 - 1}\right)^4.
	\]
\end{lemma}
Since, by our computational experiments we observed that 
the algorithm ${\cal A}_g$ has an average linear time complexity, it must expand a fraction $\Theta(n^{-1})$ of the edges.
\begin{corollary} 
	\label{cor:tlinear}
	Let $\alpha>0$ be a constant and define
	\[
	\delta_n := \sqrt{2} - \alpha\, n^{-1/4}.
	\]
	Then, under the Euclidean distribution setting for random instances, the average-case complexity of \ALG($\delta_n$) satisfies $\bar T_{\ALG(\delta_n)}(n) = \Theta(n)$.
\end{corollary}
\myproof{
	Since, by Lemma \ref{lem:maxd1}, it is $\probd{}{C > \delta_n} = \Theta(n^{-1})$, by Lemma \ref{lem:time} the average-time complexity of \ALG($\delta_n$) is $\Theta(n^{2-1})=\Theta(n)$.
}

Instead of discussing  how the setting of the constant $\alpha$
affects the probability of having good moves, in the following we find it is easier to rewrite 
\[
\delta_n = \sqrt{2} - (5\sqrt{2}) \lambda\, n^{-1/4}
\]
for a constant $\lambda>0$ and discuss the constant $\lambda$.

We first recall a very basic property of tours on Euclidean instances.
\begin{lemma}
	\label{lem:cross1}
	Let $T=(P_1,\ldots,P_n)$ be a tour on a Euclidean instance, where each $P_i$ is a point in the plane. Assume that edges $\{P_i,P_{i+1}\}$ and $\{P_j,P_{j+1}\}$ cross. Then the \opt{2} move $\mu(i,j)$ has value $>0$. Furthermore if 
	$\min\{c(P_i,P_{i+1}),c(P_j,P_{j+1})\} > l > u > \max\{c(P_i,P_j),c(P_{i+1},P_{j+1})\}$, 
	then  $\Delta(\mu(i,j)) > 2(l-u)$.
\end{lemma}
\myproof{
	Consider the following figure
	
	\begin{center}
		{\footnotesize
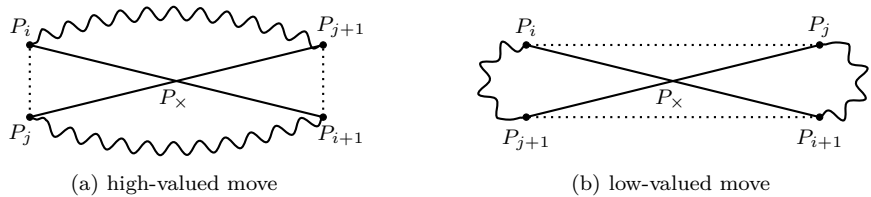
			\begin{tikzpicture}[scale=0.6]			
				\draw [ thick, snake it]   (5,0) to[out=-160,in=-20] (-1.5,0);
				\draw [ thick, snake it]   (5,1.6) to[out=160,in=20] (-1.5,1.6);
				\node at (5.4,-0.4) {$P_{i+1}$};
				\node at (-1.7,-0.4) {$P_j$};
				\draw [dotted,thick] (5,0) -- (5,1.6);
				\node at (5.4,2) {$P_{j+1}$};
				\node at (-1.7,2) {$P_i$};
				\node at (1.7,0.4) {$P_{\times}$};
				\draw [dotted,thick] (-1.5,0) -- (-1.5,1.6);
				\filldraw (5,0) circle (2pt);
				\filldraw (-1.5,0) circle (2pt);
				\filldraw (5,1.6) circle (2pt);
				\filldraw (-1.5,1.6) circle (2pt);
				\draw  [ thick](5,0) -- (-1.5,1.6);
				\draw  [ thick](-1.5,0) -- (5,1.6);
				
				\node at (1.7,-1.5) {(a) high-valued move};
				\resetorigin{11}{0}
				
				\draw[ thick, snake it] (5,0) arc (-90:90:0.8);
				\node at (5,-0.5) {$P_{i+1}$};
				\node at (-1.5,-0.5) {$P_{j+1}$};
				\draw[ thick, snake it] (-1.5,1.6) arc (90:270:0.8);
				\draw [dotted,thick] (5,0) -- (-1.5,0);
				\node at (5,2) {$P_j$};
				\node at (-1.5,2) {$P_i$};
				\node at (1.7,0.4) {$P_{\times}$};
				\draw [dotted,thick] (5,1.6) -- (-1.5,1.6);
				\filldraw (5,0) circle (2pt);
				\filldraw (-1.5,0) circle (2pt);
				\filldraw (5,1.6) circle (2pt);
				\filldraw (-1.5,1.6) circle (2pt);
				\draw  [thick](5,0) -- (-1.5,1.6);
				\draw  [thick](-1.5,0) -- (5,1.6);
				\node at (1.7,-1.5) {(b) low-valued move};
			\end{tikzpicture}
		}
	\end{center}
	
	Let $P_{\times}$ be the point in which $\{P_i,P_{i+1}\}$ and $\{P_j,P_{j+1}\}$ intersect.
	By triangle inequality $||P_i-P_{\times}|| + ||P_{\times}-P_j|| > ||P_i-P_j||=:c(P_i,P_j)$ and  $ ||P_{i+1} - P_{\times}|| + ||P_{\times} - P_{j+1}|| > ||P_{i+1},P_{j+1}||=:c(P_{i+1},P_{j+1})$, so that
	\begin{align*}
		\Delta(\mu(i,j)) & = c(P_i,P_{i+1}) + c(P_j,P_{j+1}) - c(P_i,P_j) - c(P_{i+1},P_{j+1}) \\
		& = ||P_i-P_{\times}|| + ||P_{\times}-P_{i+1}|| + ||P_j-P_{\times}|| + ||P_{\times}-P_{j+1}|| - c(P_i,P_j) - c(P_{i+1},P_{j+1})\\
		& > 0.
	\end{align*}
	The second part of the claim is obvious.
}

Notice that a pair of crossing edges implies an improving move, but the move's value could be high or not so high, depending on how small or large the angle $\widehat{P_i P_{\times}  P_j}$ is. In the previous figure, left, the angle is small and the move has a high value, while it is less so in the figure on the right.

Now we want to describe the specific type of good moves
that we will use for the analysis. Consider Figure \ref{fig:dcross1}, showing the unit square which has been divided into $(n^{1/4}/\lambda)\times (n^{1/4}/\lambda)=\sqrt{n}/\lambda^2$ squares, each of side $\lambda n^{-1/4}$. Four  of these squares are special, and they are labeled $A_1,A_2$ and $B_1,B_2$. An instance is a complete graph $K_n$ made of $n$ points and all line segments between them.
In the figure, we show some of  the points and  edges.

\begin{figure}[H]
	\caption{Explaining D-edges and D-crosses}
	\label{fig:dcross1}
	{\begin{center}
			\begin{tikzpicture}[scale=5.6]
				\fill[lightgray] (-1/6,2/6) -- (-2/6,2/6) -- (-2/6,1/2) -- (-1/6,1/2);
				\fill[lightgray] (1/6,-2/6) -- (2/6,-2/6) -- (2/6,-1/2) -- (1/6,-1/2);
				\fill[lightgray] (-1/2,1/6) -- (-1/2,2/6) -- (-2/6,2/6) -- (-2/6,1/6);
				\fill[lightgray] (1/2,-1/6) -- (1/2,-2/6) -- (2/6,-2/6) -- (2/6,-1/6);
				
				
				\draw (-0.5,-0.5) -- (1/2,-1/2) -- (1/2,1/2) -- (-1/2,1/2) -- (-1/2,-1/2);

				\draw [dotted](-1/2,-2/6) -- (1/2,-2/6);
				\draw [dotted](-1/2,-1/6) -- (1/2,-1/6);
				\draw [dotted](-1/2,0) -- (1/2,0);
				\draw [dotted](-1/2,2/6) -- (1/2,2/6);
				\draw [dotted](-1/2,1/6) -- (1/2,1/6);
				
				\draw [dotted](-2/6,-1/2) -- (-2/6,1/2);
				\draw [dotted](-1/6,-1/2) -- (-1/6,1/2);
				\draw [dotted](0,-1/2) -- (0,1/2);
				\draw [dotted](2/6,-1/2) -- (2/6,1/2);
				\draw [dotted](1/6,-1/2) -- (1/6,1/2);

				\draw [dotted](1/2,-1/2) -- (-1/2,1/2);
				\node at (0.62,-0.415) {$\lambda n^{-1/4}$};
				\draw [->] (0.55,-0.39) -- (0.55,-0.33);
				\draw [->] (0.55,-0.45) -- (0.55,-0.5);
				\node at (-1/4,5/12) {$A_1$};
				\node at (1/4,-5/12) {$A_2$};

				\node at (5/12,-1/4) {$B_2$};
				\node at (-5/12,1/4) {$B_1$};

				
				
				
				
				\tikzset{shift={(-1/2,-1/2)}}
				\node at (0.22,0.86)   [circle, fill, inner sep =1.5pt]{};
				\node at (0.69,0.06) [circle, fill, inner sep =1.5pt]{};
				\node at (0.11,0.79) [circle, fill, inner sep =1.5pt]{};
				\node at (0.89,0.21) [circle, fill, inner sep =1.5pt]{};
				
				\draw [ultra thick](0.22,0.86) -- (0.69,0.06) ;
				\draw [ultra thick](0.11,0.79) -- (0.89,0.21) ;
				
				\node at (0.73,0.7) [circle, fill, inner sep =1pt]{};
				\node at (0.33,0.74) [circle, fill, inner sep =1pt]{};
				\node at (0.44,0.24) [circle, fill, inner sep =1pt]{};
				\node at (0.23,0.54) [circle, fill, inner sep =1pt]{};
				\node at (0.55,0.34) [circle, fill, inner sep =1pt]{};
				\node at (0.78,0.44) [circle, fill, inner sep =1pt]{};
				\node at (0.62,0.51) [circle, fill, inner sep =1pt]{};
				\node at (0.19,0.39) [circle, fill, inner sep =1pt]{};
				\node at (0.08,0.12) [circle, fill, inner sep =1pt]{};
			\end{tikzpicture}
	\end{center}}
\end{figure}
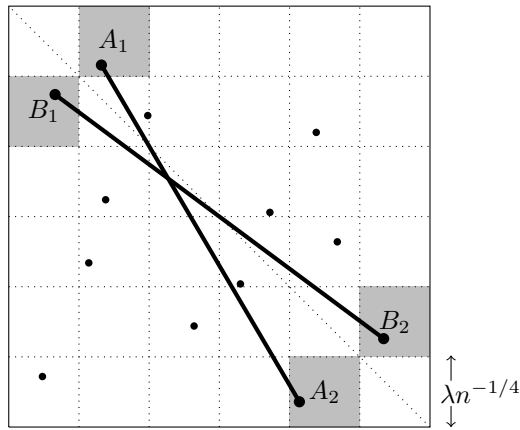

Call {\em D-edge} (for Diagonal-edge) an edge of $K_n$ which is either $A_1\between A_2$ or $B_1\between B_2$. 
Furthermore, call  {\em D-cross} (for Diagonal-cross) a pair of edges, one of which is 
$A_1\between A_2$ and the other is $B_1\between B_2$. 
Finally, call {\em C-edge} (for Corner-edge) an
edge whose endpoints are both in $A_i\cup B_i$ for $i=1,2$. Intuitively, D-edges are \myquote{long} and C-edges are \myquote{short}.

In a random instance, the tour is identified by a random permutation $\pi=(\pi_1,\pi_2,\ldots,\pi_n)$ of the nodes of $K_n$.
For each D-cross contained in the tour, if the D-cross is traversed in the right order
then there is a move which can replace two D-edges with two C-edges. 
For instance, a right order for the D-cross in the figure would be if the node in $A_1$ is labeled $\pi_i$, that in $A_2$ is $\pi_{i+1}$,  the node in $B_1$ is $\pi_j$ and that in $B_2$ is $\pi_{j+1}$, for some $i$ and  $j$. We denote this order as $\cross{A_1}{A_2}{B_1}{B_2}$. In the analysis, we consider the specific type of good moves that replace a D-cross with two C-edges, which we call {\em D-uncrossing} moves.
Since each D-edge is long at least $\sqrt{2} - 3\sqrt{2}\lambda n^{-1/4}$ and each C-edge is long at most $2\sqrt{2}\lambda n^{-1/4}$, by Lemma \ref{lem:cross1} these moves would have value greater than
\begin{equation}
	2(\sqrt2 -3\sqrt{2} \lambda n^{-1/4}) -4\sqrt{2}\lambda n^{-1/4}  = 2(\sqrt2 -5\sqrt{2} \lambda n^{-1/4}) = 2 \delta_n,
	\label{eq:valmov1}
\end{equation}
i.e., they are in fact good moves.

\begin{lemma}
	Under the Euclidean distributional setting, 
	consider the random points $P(1),\ldots,P(n)$ and the random tour $\pi=(\pi_1,\ldots,\pi_n)$. Then for each $i$, 
	\[
	\prob{(P(\pi_i)\in A_1)\land (P(\pi_{i+1})\in A_2)} = \lambda^4/n.
	\]
	\label{lem:pdedge}
\end{lemma}

\myproof{
	The probability of a point drawn at random to fall in $A_k$, for $k=1,2$, is $\lambda^2/\sqrt{n}$. Since the points $P(k)$ were drawn independently of each other, the conclusion follows.
}

\begin{lemma}
	Under the Euclidean distributional setting, 
	consider  the random points $P(1),\ldots,P(n)$, and  the random tour $\pi=(\pi_1,\ldots,\pi_n)$. Let ${\cal E}_A$ be the event \myquote{$\pi$ does not traverse any
		$A_1\rightarrow A_2$ D-edge }. Then
	\[
	\prob{{\cal E}_A} \le \left(1-\frac{\lambda^4}{n}\right)^{n/2}.
	\]
	\label{cor:pnotdedge}
\end{lemma}
\myproof{
	For each $i$ let us consider the event $D_i:=(P(\pi_i)\in A_1) \land (P(\pi_{i+1})\in A_2)$.
	By Lemma \ref{lem:pdedge}, it is $\prob{\lnot D_i}=1- \lambda^4/n$ for each $i$.
	Furthermore, it is ${\cal E}_A=\lnot D_1\land \lnot D_2 \land \cdots \land \lnot D_n$.
	
	Let us look at the odd-indexed edges of the tour, i.e., edges $(P(\pi_i),P(\pi_{i+1}))$ for $i=1,3,5,\ldots$. Since these edges are disjoint, the events $D_1,D_3,D_5,\ldots$ are independent.
	The probability that none of them occurs is 
	\[
	\prob{\lnot D_1 \land \lnot D_3 \land \cdots} = (1-\lambda^4/n)^{n/2}. 
	\]
	Since ${\cal E}_A \implies \lnot D_1 \land \lnot D_3 \land \cdots$, the result follows.
}

\begin{corollary}
	Under the Euclidean distributional setting, 
	consider  the random points $P(1),\ldots,P(n)$ and the random tour $\pi=(\pi_1,\ldots,\pi_n)$. Let ${\cal E}_{AB}$ be the event \myquote{$\pi$ does not contain any D-cross $\cross{A_1}{A_2}{B_1}{B_2}$}. Then
	\[
	\prob{{\cal E}_{AB}} \le 2(1-\lambda^4/n)^{n/2}.
	\]
\end{corollary}

\myproof{
	Let ${\cal E}_A$ be the event \myquote{$\pi$ does not traverse any $A_1\rightarrow A_2$ D-edge} and 
	${\cal E}_B$ be the event \myquote{$\pi$ does not traverse any  $B_1\rightarrow B_2$ D-edge}. 
	We have ${\cal E}_{AB} = {\cal E}_A \lor {\cal E}_B$ . Furthermore, $\prob{{\cal E}_A}=\prob{{\cal E}_B}$  and so
	\[
	\prob{{\cal E}_{AB}} = \prob{{\cal E}_A} + \prob{{\cal E}_B} - \prob{{\cal E}_A \cap {\cal E}_B} \le 2\, \prob{{\cal E}_A}.
	\]
	Then the conclusion follows from Lemma \ref{cor:pnotdedge}.
}

\begin{theorem}
	\label{cor:pgood}
	For each $p\in[0,1)$ there exist a value for $\lambda>0$
	and an integer $n_0$ such that, for each $n\ge n_0$ and  Euclidean
	random instance of size $n$, it is $\prob{\textrm{The instance is good }}  > p$.
	\label{teo:puncross}
\end{theorem}

\myproof{
	Let $\lambda$ be such that $2/\sqrt{e^{\lambda^4}} < 1-p$,  i.e.,  
	\[
	\lambda > \sqrt[4]{ 2\ln\left(\frac{2}{1-p}\right)}.
	\]
	It is
	\begin{align*}
		\lim_{n\rightarrow \infty} \prob{\textrm{there are D-uncrossings}} & \ge 1- \lim_{n\rightarrow \infty} 2(1-\lambda^4/n)^{n/2}  \\
		& = 1 - 2/\sqrt{e^{\lambda^4}} \\
		& > p
	\end{align*}
	and therefore, from some $n_0$ on, it is $\prob{\textrm{ there are D-uncrossings }} > p$.
	Since the D-uncrossings are good moves, the conclusion follows.
}

Finally, the following theorem bounds probabilistically the complexity of ${\cal A}_g$ on Euclidean instances via linear functions.
\begin{theorem}
	\label{teo:avgE}
	Consider the Euclidean distributional setting for random instances. Then, for each $p\in[0,1)$ there exists an algorithm \ALG($\delta_n$), 
	with  $\bar T_{\ALG(\delta_n)}(n)=\Theta(n)$, and an integer $n_0$ such that, for each $n\ge n_0$, it is
	\[
	\probd{}{T^n_{{\cal A}_g} \le T^n_{\ALG(\delta_n)}}  > p.
	\]
\end{theorem}
\myproof{
	The proof follows the exact same lines as the proof of Theorem \ref{teo:avgU}.
}

\section{Computational experiments and statistics}
\label{sec:results}

\subsection{Best move from a random tour} 

In this section we compare experimentally the complete enumeration (CE), the greedy algorithm (${\cal A}_g$), the blind algorithm (${\cal A}_b$), and  the algorithm $\ALG(\delta_n)$ by looking at how many moves they evaluate on average over 1000 runs. In particular, 
we generate 100 random instances and for each of them we generate 10 random tours on which we determine the best 2-OPT move.
Given that the bookkeeping costs are dominated by the number of evaluations, looking at the number of moves which are evaluated gives a pretty precise idea of the ratios between the running times as well.

\paragraph{Uniform instances.}

	\begin{table}[h]
	\begin{center}
	\begin{tabular}[h]{|r|r|r|r|r|r|c|r|r|c}
			\hline
			\rule[-2.5mm]{0mm}{0.65cm}
			$n$ &  CE & ${\cal A}_g$ & ${\cal A}_b$ & \ALG($\delta_n$) & $\bar f(n)$ & $\frac{CE}{{\cal A}_g}$ & $\frac{{\cal A}_b}{{\cal A}_g}$\\
			\hline
			2,000 &  1,999,000 &   106,462 &   133,049 &   170,190 &   169,447 & 18.78 & 1.25\\
			4,000 &  7,998,000 &   304,987 &   377,135 &   481,079 &   479,629 & 26.22 & 1.24\\
			6,000 & 17,997,000 &   560,647 &   682,285 &   881,361 &   881,356 & 32.10 & 1.22  \\
			8,000 & 31,996,000 &   871,001 &  1,057,508 &  1,355,070 &  1,357,106 & 36.73 & 1.21\\
			10,000 & 49,995,000 &  1,201,409 &  1,468,602 &  1,886,631 &  1,896,756 & 41.61 & 1.22 \\
			12,000 & 71,994,000 &  1,567,947 &  1,918,657 &  2,497,627 &  2,493,476 & 45.92 & 1.22\\
			14,000 & 97,993,000 &  1,986,524 &  2,425,136 &  3,137,567 &  3,142,251 & 49.33 & 1.22 \\
			16,000 & 127,992,000 &  2,453,347 &  2,981,419 &  3,857,294 &  3,839,197 & 52.17 & 1.22\\
			18,000 & 161,991,000 &  2,910,420 &  3,525,385 &  4,576,299 &  4,581,189 & 55.66 & 1.21 \\
			20,000 & 199,990,000 &  3,368,334 &  4,111,101 &  5,345,312 &  5,365,643 & 59.37 & 1.22 \\
			22,000 & 241,989,000 &  3,963,375 &  4,779,938 &  6,206,599 &  6,190,371 & 61.06 & 1.21\\
			24,000 & 287,988,000 &  4,486,287 &  5,465,894 &  7,052,514 &  7,053,497 & 64.19 & 1.22\\
			\hline
		\end{tabular}
		\end{center}
\caption{Average number of moves evaluated for finding the best move on a random tour. Results for uniform instances.}
	\label{tab:1moveUni}
	\end{table}

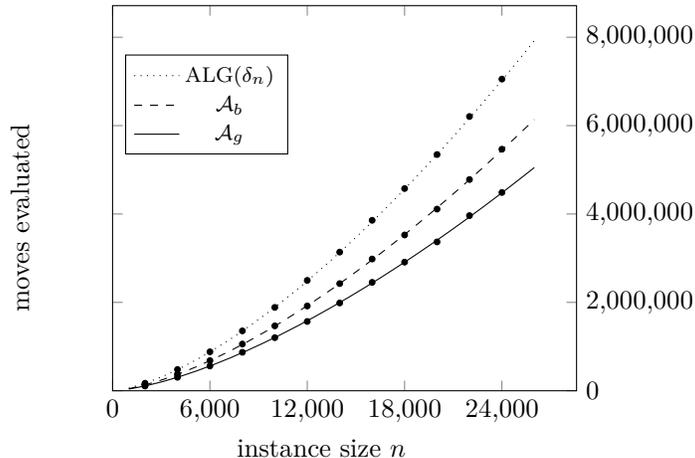
\begin{figure}[tbp]
	\caption{Plots of possible fittings of the average time complexity for finding the  best move on a random tour (uniform instances). \ALG($\delta_n$),  fit $\alpha n^{\frac{3}{2}}= 1.89\, n^{\frac{3}{2}}$. ${\cal A}_b$, fit is $1.46\, n^{\frac{3}{2}}$. ${\cal A}_g$,fit is $1.20\, n^{\frac{3}{2}}$. 
		Each dot corresponds to the average over 1000 trials of size $n$. \comment{altri possibili fit: \color{blue} cyan, $0.675\, n^{\frac{3}{2}}\, \log n^{\frac{1}{4}}$; blue, $0.119\, n^{\frac{3}{2}}\, \log n$} 
	}
	\label{fig:plotUni}
	{\begin{center}
			\begin{tikzpicture}[scale=0.9]
				\pgfplotsset{scaled x ticks=false}
				\pgfplotsset{set layers}
				\begin{axis}[ 
					scaled y ticks = false, ymin=0, xmin=0,
					xlabel=instance size $n$,
					xtick={0,6000,12000,18000,24000},
					ylabel=moves evaluated,
					yticklabel pos=right,
					yticklabel style = { /pgf/number format/fixed },
					legend style={font=\footnotesize, 
						legend pos =north west},				]
					
					\addplot[ domain=1000:26000, color=black, dotted ] {(1.89*x^1.5)};
					\addlegendentry{\ALG($\delta_n$)}	
					\addplot[ domain=1000:26000, color=black,dashed ] {(1.464*x^1.5)};
					\addlegendentry{${\cal A}_b$}	
					\addplot[ domain=1000:26000, color=black ] {(1.205*x^1.5)};
					\addlegendentry{${\cal A}_g$}	
					\addplot[ color=black, only marks, /tikz/mark size={1.2pt} ] coordinates{ 
						( 2000,   106462 )
						(4000,  304987 )
						(6000,  560647)
						(8000,   871001) 
						(10000,  1201409)
						(12000 , 1567947)
						(14000, 1986524)
						(16000,  2453347)
						(18000,  2910420) 
						(20000, 3368334)
						(22000, 3963375)
						(24000, 4486287)
					};	
					\addplot[ color=black, only marks, /tikz/mark size={1.2pt} ] coordinates{
						( 2000,  133049)
						(4000,   377135)
						(6000,  682285)
						(8000,  1057508) 
						(10000, 1468602)
						(12000,  1918657)
						(14000, 2425136)
						(16000, 2981419)
						(18000,  3525385) 
						(20000, 4111101) 
						(22000, 4779938)
						(24000,5465894)
					};

					\addplot[ color=black, only marks, /tikz/mark size={1.2pt} ] coordinates{
						( 2000,  170190)
						(4000,   481079) 
						(6000,   881361)
						(8000, 1355070)
						(10000,  1886631)
						(12000,  2497627)
						(14000,  3137567)
						(16000,  3857294)
						(18000 ,  4576299)
						(20000, 5345312)
						(22000,  6206599)
						(24000, 7052514)
					};			

				\end{axis}l
			\end{tikzpicture}
			
		\end{center}
	}
\end{figure}

In Table \ref{tab:1moveUni} we report 
the average number of moves (rounded to integer) evaluated by
CE, ${\cal A}_g$, ${\cal A}_b$, and   $\ALG(\delta_n)$.
In the experiment,  we set $\delta_n:=1 - 1.89/\sqrt{n}$, where  $\alpha\simeq 1.9$ was chosen, after a little 
tuning, since it is a value large enough to guarantee a good probability of no errors. Indeed, out of 12,000 instances considered, the algorithm \ALG($\delta_n$) {\em always} found the best move.

Column $\bar f(n)$ reports the theoretical expected complexity of \ALG($\delta_n$), i.e., $\alpha(n-3)n^{1/2}$. It can be seen how this column is very close to the actual averages observed in the experiments.
Column $\frac{CE}{{\cal A}_g}$ gives the ratio between the number of moves of complete enumeration and of the greedy algorithm, showing that we can achieve speed-ups of up to $60\times$ over instances of size $\le 24000$. The table also reports the ratios between the number of moves of the blind  and of the greedy algorithm (column $\frac{{\cal A}_b}{{\cal A}_g}$). Since this ratio stays pretty much constant (around 1.22) it appears from the table that ${\cal A}_g$ and ${\cal A}_b$ both evaluate a sub-quadratic number of moves of the same asymptotic growth but with a smaller constant for ${\cal A}_g$. 
In Figure \ref{fig:plotUni} we have plotted the same values and we have fitted the dots with functions $\Theta(n^{3/2})$, namely $\alpha n^{3/2}$ for \ALG($\delta_n$), and (estimated by power regression)  
$1.20\, n^{\frac{3}{2}}$ for the greedy algorithm and $1.46\, n^{\frac{3}{2}}$ 
for the blind algorithm.

\paragraph{Euclidean and TSPLIB geometric instances.}

We have performed a similar set of experiments on random Euclidean instances. In Table \ref{tab:1moveEuc} we can see that the greedy
algorithm is from {\em two to three orders of magnitude}  faster than
complete enumeration when looking for the best move on a random tour on graphs with up to 24,000 nodes.
The values are averages over 1,000 experiments for each size $n$, 
exactly as before. The blind algorithm exhibits a similar time complexity,
but it is roughly 2.5 times slower than greedy. 
The fixed threshold algorithm has been run with
$\delta_n= \sqrt{2} - 2.5/\sqrt[4]{n}$, corresponding to
$\lambda=1/(2\sqrt2)$. 
We remark that the algorithm \ALG($\delta_n)$ found the optimal move, over {\em all }the 12,000 trials.
In Figure \ref{fig:plotEuc} we can see the data of Table \ref{tab:1moveEuc} plotted in a graph. The linear behavior of ${\cal A}_g$ and ${\cal A}_b$ can be appreciated by looking at the interpolating functions, respectively, $7.7\,n$ and $19.3\,n$.

\begin{table}[h]
\begin{center}
		\begin{tabular}[h]{|r|r|r|r|r|r|c|r|r|c}
			\hline
			\rule[-2.5mm]{0mm}{0.65cm}
			$n$ &  CE & ${\cal A}_g$ & ${\cal A}_b$ & \ALG($\delta_n$) & $\bar f(n)$ & $\frac{CE}{{\cal A}_g}$ & $\frac{{\cal A}_b}{{\cal A}_g}$\\
			\hline
			2,000 &  1,999,000 &    15,786 &    34,874 &    63,626 &    63,100 & 126.63 & 2.21\\
			4,000 &  7,998,000 &    32,811 &    73,709 &   119,614 &   119,846 & 243.76 & 2.25\\
			6,000 & 17,997,000 &    46,710 &   110,513 &   173,905 &   176,063 & 385.29 & 2.37 \\
			8,000 & 31,996,000 &    61,073 &   149,504 &   232,322 &   231,907 & 523.90 & 2.45\\
			10,000 & 49,995,000 &    78,926 &   189,164 &   285,921 &   287,487 & 633.44 & 2.40 \\
			12,000 & 71,994,000 &    93,552 &   227,500 &   340,711 &   342,869 & 769.56 & 2.43\\
			14,000 & 97,993,000 &   110,450 &   269,611 &   397,165 &   398,093 & 887.21 & 2.44 \\
			16,000 & 127,992,000 &   124,632 &   308,551 &   457,667 &   453,189 & 1026.95 & 2.48\\
			18,000 & 161,991,000 &   141,852 &   350,996 &   513,313 &   508,177 & 1141.97 & 2.47 \\
			20,000 & 199,990,000 &   156,056 &   386,363 &   569,531 &   563,072 & 1281.52 & 2.48\\
			22,000 & 241,989,000 &   169,574 &   423,487 &   618,985 &   617,886 & 1427.03 & 2.50\\
			24,000 & 287,988,000 &   181,513 &   464,918 &   675,163 &   672,629 & 1586.59 & 2.56\\
			\hline
		\end{tabular}
\end{center}		
\caption{Average number of moves evaluated for finding the best move on a random tour. Results for  Euclidean instances.}
		\label{tab:1moveEuc}
\end{table}

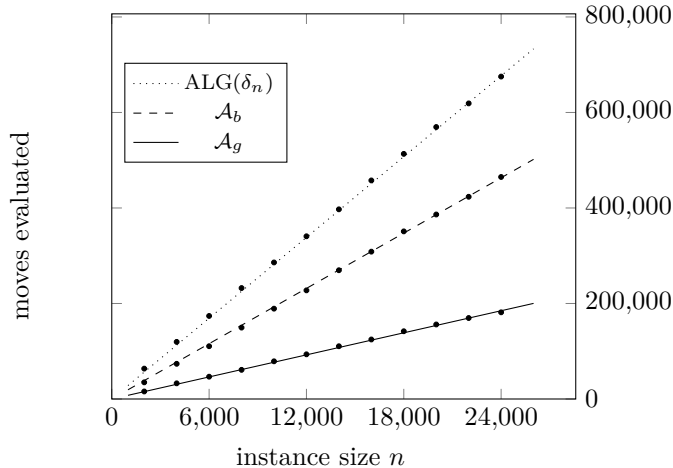
\begin{figure}[tbp]
	\caption{Plots of possible fittings of the average time complexity for finding the  best move on a random tour (Euclidean instances). \ALG($\delta_n$):  fit is $28.2\, n$. ${\cal A}_b$: fit is $19.3\, n$. ${\cal A}_g$: fit is $7.7\, n$. 
		Each dot corresponds to the average over 1000 trials of size $n$.
	}
	\label{fig:plotEuc} 	
	{\begin{center}
			\begin{tikzpicture}[scale=0.9]
				\pgfplotsset{scaled x ticks=false}
				\pgfplotsset{set layers}
				\begin{axis}[ 
					scaled y ticks = false, ymin=0, xmin=0,
					xlabel=instance size $n$,
					xtick={0,6000,12000,18000,24000},
					ylabel=moves evaluated,
					yticklabel pos=right,
					yticklabel style = { /pgf/number format/fixed },
					legend style={font=\footnotesize, 
						legend pos =north west},
					]
					\addplot[ domain=1000:26000, color=black, dotted ] {(28.2*x)};
					\addlegendentry{\ALG($\delta_n$)}	
					\addplot[ domain=1000:26000, color=black,dashed ] {(19.3*x)};
					\addlegendentry{${\cal A}_b$}	
					\addplot[ domain=1000:26000, color=black ] {(7.7*x)};
					\addlegendentry{${\cal A}_g$}
					\addplot[ color=black, only marks, /tikz/mark size={1pt} ] coordinates{
						(2000,    63626) 
						(4000,  119614 )
						(6000 ,   173905 )
						(8000,   232322 )
						(10000,  285921)
						(12000,  340711) 
						(14000 ,  397165 )
						(16000,   457667)
						(18000,   513313)
						(20000, 569531)
						(22000,   618985) 
						(24000, 675163) 
					};
					\addplot[ color=black, only marks, /tikz/mark size={1pt} ] coordinates{
						( 2000,   15786)
						(4000,   32811)
						(6000,   46710 )
						(8000,   61073)
						(10000,    78926) 
						(12000,    93552)
						(14000,  110450)
						(16000,  124632 )
						(18000,   141852 )
						(20000,   156056 )
						(22000,  169574 )
						(24000,  181513) 
					};
					\newcommand\xM{30000}
					\newcommand\eXP{1.333}
					\newcommand\eXB{1.5}
					\addplot[ color=black, only marks, /tikz/mark size={1pt} ] coordinates{
						( 2000,   34874) 
						(4000,   73709) 
						(6000, 110513 )
						(8000,  149504 )
						(10000,   189164)(
						12000,  227500)
						(14000,  269611)
						(16000,   308551) 
						(18000,  350996 )
						(20000,  386363 )
						(22000, 423487) 
						(24000,  464918)
					};
					
				\end{axis}
			\end{tikzpicture}
			
		\end{center}
	}
\end{figure}

Some test-bed instances on the repository TSPLIB \mycite{tsplib} are of geometric nature, and we have tested our algorithm on those as well. In particular, there are some Euclidean instances (but they are not random, they correspond to some networks of world cities), and other are metric, not Euclidean, instances. We have selected the largest such instances (with the exception of {\tt pla85900}, that, with $\ge 85,900$ nodes, was too big for our computer setting).
The results are reported in Table \ref{tab:TSPLIB1}.
It can be seen that our method achieves a speed-up of two to three orders of magnitude in finding the best move on a random tour.

\begin{table}[h]
\begin{center}
\begin{tabular}[h]{|l|r|r|r|r|c|r|r|c}
			\hline
			\rule[-2.5mm]{0mm}{0.65cm}
			name &	$n$ &  ${\cal A}_g$ & ${\cal A}_b$ & CE & $\frac{CE}{{\cal A}_g}$ & $\frac{{\cal A}_b}{{\cal A}_g}$\\
			\hline			
			{\tt euc2d/rl5915} & 5,915 & 59,258 & 113,747 & 17,490,655 & 295.16 & 1.92 \\
			{\tt euc2d/rl5934} & 5,934 & 51,261 & 101,539 & 17,603,211 & 343.40 & 1.98 \\
			{\tt ceil2d/pla7397} & 7,397 &  48,665 &  127,695 & 27,354,106 & 562.08 & 2.62 \\
			{\tt euc2d/rl11849} &  11,849 & 98,457  & 223,416  & 70,193,476 & 712.93 & 2.27 \\
			{\tt euc2d/usa13509} & 13,509 & 104,147  & 232,979  & 91,239,786 & 876.06 & 2.24 \\
			{\tt euc2d/brd14051} & 14,051 & 170,286  & 345,863  & 98,708,275 & 579.66 & 2.03 \\
			{\tt euc2d/d15112} & 15,112  & 195,385  & 375,611  & 114,178,716 & 584.37 & 1.92 \\
			{\tt euc2d/d18512} & 18,512  &  174,374 & 409,420  &  171,337,816 & 982.58 & 2.34  \\
			{\tt ceil2d/pla33810}& 33,810 & 371,561 &  832,667 & 571,541,145  & 1538.21  & 2.24  \\
			\hline
	\end{tabular}
\end{center}
\caption{Finding the best move on a random tour on TSPLIB instances.}
\label{tab:TSPLIB1}	
\end{table}

\subsection{Convergence to a local optimum}

Given that  finding the best move at the beginning of
the local search is much faster with
our algorithm than with the standard approach, we
were optimistic about the fact that the time for the whole convergence would have been much shorter as well. Unfortunately, this is not
the case, since the effectiveness of our approach decreases along the
path to the local optimum. Indeed, at some point, our procedure can in fact become 
slower than complete enumeration given that, although they are both quadratic, the multiplicative constant factor for the $n^2$ term is  worse in our procedure  (for some minor technical details, see Appendix B). 


The slow-down phenomenon can be explained as follows. When we approach the
local optimum, the value of the best move decreases dramatically. While at the beginning of the convergence there are many improving moves, and some of them have a really big value, near the local optimum most moves are worsening, and the few improving moves have a small value. 
At this point, also the
thresholds which determine which pivots are \good become pretty low and hence almost all nodes are \good\ and get expanded. 

To counter the effect that our method could become slower than complete enumeration, we propose a very elementary switch-condition to perform a full local search, namely: {\em start the search looking for the best move with our algorithm, but at each iteration count how many moves are evaluated. If at any step the number of moves evaluated is ``too large'' (in a way that we will define soon),  then switch to 
	complete enumeration for the rest of the convergence}.

\pgfplotsset{scaled y ticks=false}
\begin{figure}[ht]
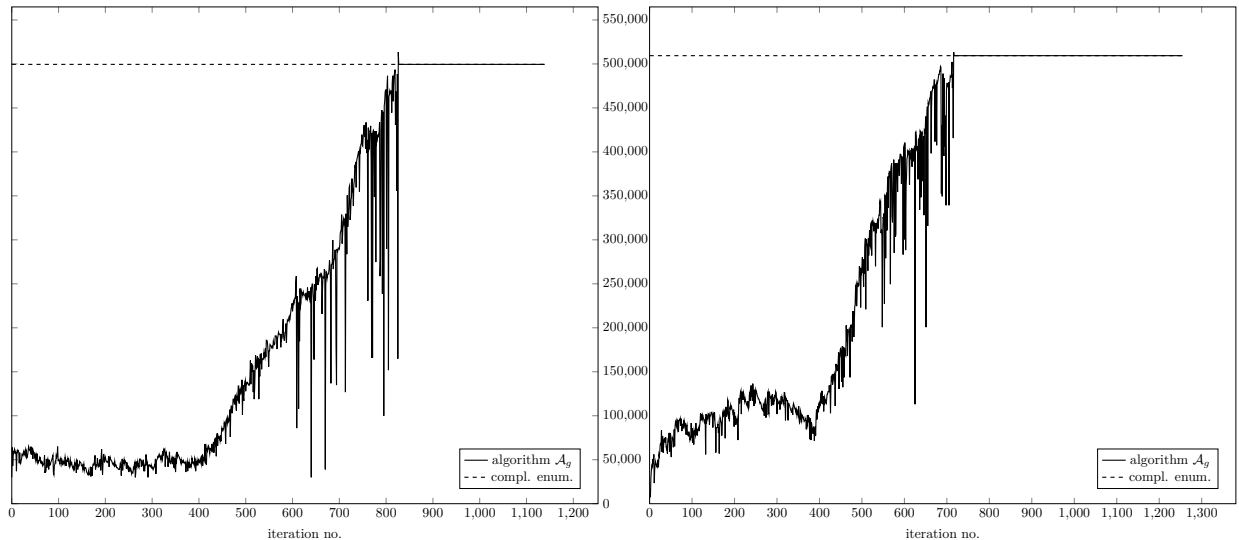

	\caption{Number of moves evaluated by CE and the greedy algorithm in the convergence to a local optimum on a random graph of $n=1000$ nodes. $x$-axis is iteration number, $y$-axis is number of moves evaluated to find the best. (Left:) Uniform  instance. (Right:) Euclidean  instance.}
	\label{fig:convergenceM}
	\begin{center}

	\end{center}
\end{figure}

A first, natural threshold for 
this switch could be simply to switch as soon as 
our algorithm evaluates as many moves as complete enumeration. For a graph of 1,000 nodes, we have 
observed empirically that this happens more or less at two thirds of the convergence. For example, in Figure \ref{fig:convergenceM} we plot the number of moves
evaluated by ${\cal A}_g$ in a local search convergence on a random graph, which took almost 1,200 moves on a uniform instance (left) and
almost 1,300 moves on a Euclidean instance (right).
For the uniform instance,  it can
be seen that until move 400 the algorithm takes more or less always the same time to find the best move, evaluating around 50,000 candidates, i.e., ten times less than complete enumeration.
From move 400 to around 800 the algorithm starts to take longer and longer
to find the best move, although, occasionally, the move is still found
by examining much fewer candidates than complete enumeration. 
In the last third of the convergence, the number of candidates 
examined by the greedy algorithm is comparable to, or exceeds, that of
complete enumeration, and so a switch is performed.

In Figure \ref{fig:convergenceM} (right) we see the same type of plot for
the convergence on a random Euclidean instance. Although finding the starting move takes fewer evaluations than for the uniform instance (i.e., 5,994 vs 22,977), soon finding the best move on the Euclidean instance becomes
somewhat harder than for the uniform one. At any rate, also in this case, until step 400 there is no significant increase in the work needed for finding the best move. At this point the trend becomes a steady increase until about move 800, when we switch to complete enumeration.

\pgfplotsset{scaled y ticks=false}
\begin{figure}[t]
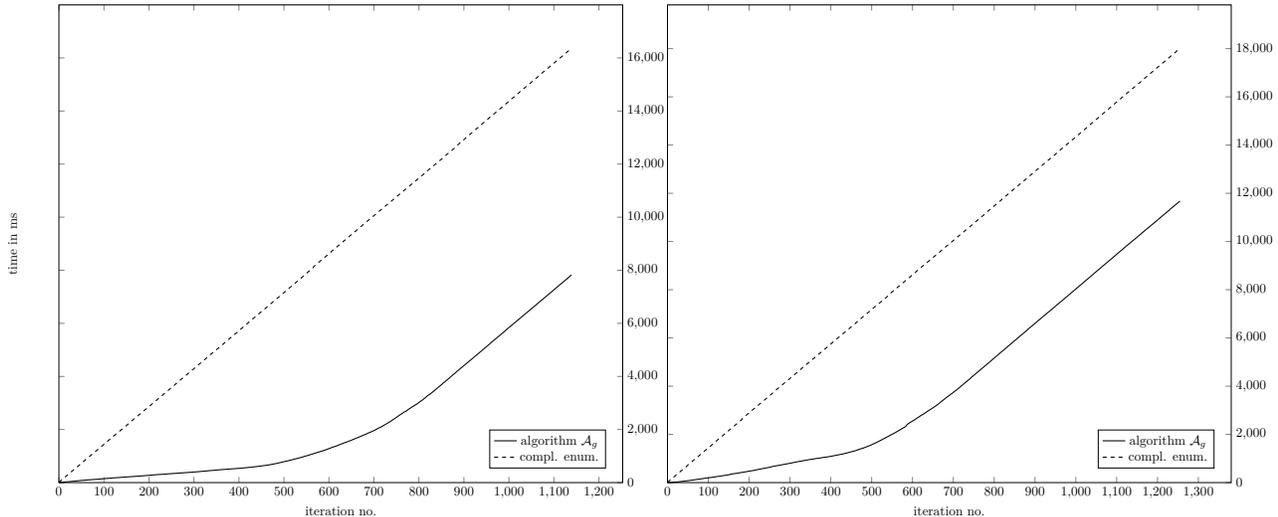

	\caption{Time (ms) taken by CE and the greedy algorithm on a convergence to a local optimum on a random graph of $n=1000$ nodes. Switch made when greedy evaluates as many moves as CE.  Left: Uniform  instance. Right: Euclidean  instance.}
	\label{fig:convergenceT}
	
	\begin{center}
		\hspace{-0.74cm}

	\end{center}
\end{figure}

As far as the running times for these two instances, in Figure \ref{fig:convergenceT} we can see how the 
greedy algorithm is initially much faster than complete
enumeration, but then its performance starts to deteriorate and eventually there is the switch. From that point on, the time complexity becomes that of
complete enumeration. Overall, there is an improvement of about 50\% for the time needed by the uniform instance and of 33\% for the
Euclidean instance.

\paragraph{Determining the best switch-point.}

The first, simple,
idea for switching from ${\cal A}_g$ to brute force was that of counting the number $c$ of moves that are evaluated
at each iteration, and switch  as soon as they are more than those evaluated by complete enumeration, i.e., $c \ge \frac{1}{2} n (n-1)$. In this paragraph we address the problem of determining if this strategy is 
indeed a good one, or if switching at some other moment could have been better. 

Notice that, since with our method there is an overhead in 
computing each move, the time to evaluate $k$ moves with our method is larger than evaluating $k$ moves with CE and so to account for this we might switch to CE when  $c \ge \beta n (n-1)$ even if  $\beta < 1/2$.
To find the best $\beta$, we have set up the following experiment. For a fixed $n$:
\begin{enumerate}
	\item We generated a random instance and a random starting permutation $\pi^0$.
	\item We ran local search with our algorithm up to the point  when our algorithm evaluates 
	$\frac{1}{2} n (n-1)$ moves. 
	Let this  iteration be $s$. This is the point when we would have 
	reverted to CE anyway, and so all our savings can only be made 
	in the iterations $1,\ldots,s$. 
	\item We ran local search with CE, starting from the same
	permutation $\pi^0$ as before up to the iteration 
	$s$. Let $t_s$ be the time taken by CE to get to iteration
	$s$. 
	\item We divided the interval $(0,1/2)$ in 5 parts, and considered setting $\beta$ to $h$ fifths of the interval. Since we quickly observed that the smallest values of $\beta$ yielded the worst results,  we limited the search of the best $\beta$ to $h=3,4,5$. For each
	such $\beta := h/10$, we ran our algorithm 
	to determine the time $t(\beta)$ that it takes to arrive
	at iteration $s$ if the switch is done as soon as
	$\beta n (n-1)$ moves are evaluated,
	as well as the total completion time $t_{tot}$. 
	
\end{enumerate}
We have repeated the above experiment 5 times for each $n=2000,4000,6000,8000$ and we obtain the following tables (Table \ref{tab:beta2000} -- \ref{tab:beta8000}).
The tables should be read as follows. Columns 
\myquote{$t_{100}$} 
(time in seconds for the first 100 steps of the convergence) and \myquote{$\textrm{impr}_{100}$} 
(improvement over CE in the first 100 steps)
are included to show the big gain of our method with respect to complete enumeration in the first stages of the convergence. Column 
\myquote{$t_{tot}$}
is the total time of the convergence and 
column 
\myquote{$\textrm{impr}_{tot}$} 
is	the overall savings over CE.
Column \myquote{$L$} is the total number of LS iterations (i.e., the length of the convergence). 
Column \myquote{$s$} is the iteration no. at which the greedy algorithm switches to CE.
\myquote{$t(\beta)$} is the time taken by ${\cal A}_g$ up to the switch iteration $s$, while \myquote{$t_s$} is the time taken by Complete Enumeration up to iteration $s$.
Finally, \myquote{avgMpI} is the average number of moves evaluated per iteration.

Each table is divided into two halves, where the first half refers to the algorithm
as we have presented it, while the second half is relative to a minor variant (discussed in Appendix B) that yields a very small improvement, on some instances. 
In the tables we have highlighted in boldface the setting yielding the best running time for each instance. From the experiments, it appears that the best setting, for instances of the size considered, would be to switch when $\beta=4/10$, corresponding to when the greedy algorithm evaluates 80\% of all possible moves. This setting is the winner for 11 out of 20 cases, for the basic greedy version, and 10 out of 20 cases for the variant we mentioned.  The overall savings are around 50\% of the time of complete enumeration. Setting $\beta=5/10$ is the best option in 6 out of 20 cases, while $\beta=3/10$ is definitely the worst option.

\begin{table}
\begin{center}
\begin{tabular}{|lccccccccr|}
			\hline
			algo & $t_{100}$ & $\textrm{impr}_{100}$ & $T_{tot}$ & $\textrm{impr}_{tot}$ & $L$ & $s$ & $t(\beta)$ & $t_s$ & avgMpI\\
			\hline
			CE & 1.321s &   0.0 & 31.38s &   0.0 &   2421 & - & - & - &  1999000\\
			$\beta=3/10$ & 0.178s & 86.6\% &  18.70s & 40.4\% &   2421 & 1411 & 5.43s & 18.28s &  1026947\\
			$\beta=4/10$ & 0.177s & 86.6\% & 19.63s & 37.4\% &   2421 & 1568 & 7.80s & 20.34s &   982218\\
			$\beta=5/10$ & 0.143s & 89.2\% & {\bf 18.64s} & {\bf 40.6\%} &   2421 & 1765 & 10.07s & 22.90s &   953043\\
			\hline
			CE & 1.301s &   0.0 & 32.47s &   0.0 &   2432 & - & - & - &  1999000\\
			$\beta=3/10$ & 0.149s & 88.6\% & {\bf 18.08s} & {\bf 44.3\%} &   2432 & 1463 & 5.49s & 19.02s &  1020521\\
			$\beta=4/10$ & 0.153s & 88.3\% & 18.18s & 44.0\% &   2432 & 1653 & 8.09s & 21.46s &   970851\\
			$\beta=5/10$ & 0.175s & 86.5\% & 18.93s & 41.7\% &   2432 & 1771 & 10.18s & 23.02s &   955734\\
			\hline
			CE & 1.379s &   0.0 & 33.34s &   0.0 &   2438 & - & - & - &  1999000\\
			$\beta=3/10$ & 0.157s & 88.6\% & {\bf 18.23s }& {\bf 45.3\%} &   2438 & 1423 & 4.99s & 19.03s &  1034536\\
			$\beta=4/10$ & 0.169s & 87.7\% & 18.46s & 44.6\% &   2438 & 1630 & 7.74s & 21.82s &   977569\\
			$\beta=5/10$ & 0.232s & 83.2\% & 20.01s & 40.0\% &   2438 & 1778 & 10.28s & 23.95s &   958197\\
			\hline
			CE & 1.270s &   0.0 & 31.81s &   0.0 &   2463 & - & - & - &  1999000\\
			$\beta=3/10$ & 0.134s & 89.4\% & {\bf 18.31s} & {\bf 42.4\% }&   2463 & 1421 & 4.92s & 18.27s &  1044769\\
			$\beta=4/10$ & 0.148s & 88.3\% & 18.57s & 41.6\% &   2463 & 1566 & 6.93s & 20.07s &  1005130\\
			$\beta=5/10$ & 0.152s & 88.0\% & 19.30s & 39.3\% &   2463 & 1745 & 9.96s & 22.40s &   982087\\
			\hline
			CE & 1.298s &   0.0 & 31.30s &   0.0 &   2409 & - & - & - &  1999000\\
			$\beta=3/10$ & 0.176s & 86.5\% & 19.13s & 38.9\% &   2409 & 1457 & 5.36s & 18.80s &  1003174\\
			$\beta=4/10$ & 0.163s & 87.5\% & {\bf 17.76s} & {\bf 43.3\% } &   2409 & 1628 & 7.47s & 21.06s &   952764\\
			$\beta=5/10$ & 0.176s & 86.5\% & 18.50s & 40.9\% &   2409 & 1792 & 10.40s & 23.16s &   931831\\
			\hline
			\hline				
			\hline			
			CE & 1.267s &   0.0 & 31.02s &   0.0 &   2421 & - & - & - &  1999000\\
			SF5-3/5 & 0.164s & 87.1\% & {\bf 17.63s} & {\bf 43.2\%} &   2421 & 1438 & 4.91s & 18.33s &  1012032\\
			SF5-4/5 & 0.150s & 88.2\% & 17.78s & 42.7\% &   2421 & 1616 & 7.06s & 20.59s &   961722\\
			SF5-5/5 & 0.270s & 78.7\% & 18.64s & 39.9\% &   2421 & 1830 & 11.08s & 23.38s &   930478\\
			\hline
			CE & 1.279s &   0.0 & 34.94s &   0.0 &   2432 & - & - & - &  1999000\\
			SF5-3/5 & 0.156s & 87.8\% & {\bf 17.63s} & {\bf 49.5\%} &   2432 & 1493 & 5.65s & 20.76s &  1003384\\
			SF5-4/5 & 0.152s & 88.1\% & 17.83s & 49.0\% &   2432 & 1670 & 8.03s & 23.50s &   955736\\
			SF5-5/5 & 0.155s & 87.9\% & 18.40s & 47.3\% &   2432 & 1813 & 10.47s & 26.10s &   936685\\
			\hline
			CE & 1.275s &   0.0 & 31.25s &   0.0 &   2438 & - & - & - &  1999000\\
			SF5-3/5 & 0.164s & 87.1\% & {\bf 17.72s} & {\bf 43.3\%} &   2438 & 1450 & 5.09s & 18.52s &  1018747\\
			SF5-4/5 & 0.161s & 87.4\% & 17.86s & 42.8\% &   2438 & 1667 & 7.96s & 21.26s &   958047\\
			SF5-5/5 & 0.159s & 87.5\% & 18.61s & 40.4\% &   2438 & 1838 & 10.89s & 23.47s &   936531\\
			\hline
			CE & 1.278s &   0.0 & 32.00s &   0.0 &   2463 & - & - & - &  1999000\\
			SF5-3/5 & 0.138s & 89.2\% & {\bf 18.09s} & {\bf 43.5\%} &   2463 & 1440 & 4.95s & 18.78s &  1031720\\
			SF5-4/5 & 0.139s & 89.1\% & 18.16s & 43.2\% &   2463 & 1624 & 7.44s & 21.14s &   982368\\
			SF5-5/5 & 0.145s & 88.6\% & 19.55s & 38.9\% &   2463 & 1758 & 9.99s & 22.88s &   963758\\
			\hline
			CE & 1.278s &   0.0 & 30.81s &   0.0 &   2409 & - & - & - &  1999000\\
			SF5-3/5 & 0.249s & 80.5\% & 21.51s & 30.2\% &   2409 & 1516 & 7.21s & 19.36s &   975039\\
			SF5-4/5 & 0.174s & 86.4\% & {\bf 17.83s} & {\bf 42.1\%} &   2409 & 1652 & 7.84s & 21.10s &   936786\\
			SF5-5/5 & 0.160s & 87.5\% & 18.01s & 41.5\% &   2409 & 1852 & 10.83s & 23.63s &   910122\\
			\hline
	\end{tabular}
\end{center}
\caption{Runs to determine the best $\beta$. Above: ${\cal A}_g$. Below: stronger ${\cal A}_g$. Five instances, $n=2000$.}
\label{tab:beta2000}
\end{table}

\begin{table}
\begin{center}
\begin{tabular}{|lccccccccr|}
			\hline
			algo & $t_{100}$ & $\textrm{impr}_{100}$ & $T_{tot}$ & $\textrm{impr}_{tot}$ & $L$ & $s$ & $t(\beta)$ & $t_s$ & avgMpI\\
			\hline
			CE & 7.406s &   0.0 & 375.68s &   0.0 &   4986 & - & - & - &  7998000\\
			$\beta=3/10$ & 0.677s & 90.9\% & 199.56s & 46.9\% &   4986 & 3069 & 56.44s & 229.64s &  3938640\\
			$\beta=4/10$ & 0.663s & 91.0\% & {\bf 197.68s} &{\bf 47.4\%} &   4986 & 3509 & 86.25s & 263.10s &  3704082\\
			$\beta=5/10$ & 0.653s & 91.2\% & 198.04s & 47.3\% &   4986 & 3821 & 108.90s & 286.97s &  3614730\\
			\hline
			CE & 7.428s &   0.0 & 378.77s &   0.0 &   5031 & - & - & - &  7998000\\
			$\beta=3/10$ & 0.690s & 90.7\% & 206.21s & 45.6\% &   5031 & 3023 & 54.66s & 227.25s &  4012509\\
			$\beta=4/10$ & 0.616s & 91.7\% & {\bf 199.83s} & {\bf 47.2\%} &   5031 & 3439 & 79.72s & 258.22s &  3786022\\
			$\beta=5/10$ & 0.634s & 91.5\% & 202.53s & 46.5\% &   5031 & 3692 & 102.32s & 277.05s &  3727908\\
			\hline
			CE & 7.444s &   0.0 & 373.32s &   0.0 &   4977 & - & - & - &  7998000\\
			$\beta=3/10$ & 0.661s & 91.1\% & 198.35s & 46.9\% &   4977 & 3077 & 55.02s & 230.38s &  3919739\\
			$\beta=4/10$ & 0.648s & 91.3\% & {\bf 195.83s} & {\bf 47.5\%} &   4977 & 3377 & 75.92s & 252.77s &  3760499\\
			$\beta=5/10$ & 0.646s & 91.3\% & 199.62s & 46.5\% &   4977 & 3681 & 101.06s & 275.37s &  3677210\\
			\hline
			CE & 7.594s &   0.0 & 379.57s &   0.0 &   5062 & - & - & - &  7998000\\
			$\beta=3/10$ & 0.641s & 91.6\% & 197.88s & 47.9\% &   5062 & 3023 & 49.27s & 225.96s &  4029621\\
			$\beta=4/10$ & 0.643s & 91.5\% & {\bf 194.83s }& {\bf 48.7\%} &   5062 & 3353 & 70.62s & 250.52s &  3855573\\
			$\beta=5/10$ & 0.619s & 91.8\% & 200.40s & 47.2\% &   5062 & 3748 & 101.70s & 281.25s &  3750503\\
			\hline
			CE & 7.156s &   0.0 & 365.01s &   0.0 &   5028 & - & - & - &  7998000\\
			$\beta=3/10$ & 0.591s & 91.7\% & 203.34s & 44.3\% &   5028 & 3027 & 52.67s & 218.98s &  3993228\\
			$\beta=4/10$ & 0.599s & 91.6\% & 200.59s & 45.0\% &   5028 & 3413 & 78.90s & 247.73s &  3789793\\
			$\beta=5/10$ & 0.606s & 91.5\% & {\bf 198.76s} & {\bf 45.5\%} &   5028 & 3688 & 100.55s & 267.98s &  3716365\\
			\hline
			\hline				
			\hline						
			CE & 7.805s &   0.0 & 390.01s &   0.0 &   4986 & - & - & - &  7998000\\
			SF5-3/5 & 0.666s & 91.5\% & 206.84s & 47.0\% &   4986 & 3094 & 58.16s & 240.01s &  3903740\\
			SF5-4/5 & 0.664s & 91.5\% & {\bf 204.76s} & {\bf 47.5\%} &   4986 & 3556 & 91.27s & 275.95s &  3649850\\
			SF5-5/5 & 0.708s & 90.9\% & 205.94s & 47.2\% &   4986 & 3883 & 117.66s & 301.67s &  3554548\\
			\hline
			CE & 8.202s &   0.0 & 394.16s &   0.0 &   5031 & - & - & - &  7998000\\
			SF5-3/5 & 0.634s & 92.3\% & 212.59s & 46.1\% &   5031 & 3061 & 56.73s & 237.46s &  3966810\\
			SF5-4/5 & 0.643s & 92.2\% & {\bf 208.34s} & {\bf 47.1\%} &   5031 & 3470 & 83.32s & 269.28s &  3742725\\
			SF5-5/5 & 0.608s & 92.6\% & 208.76s & 47.0\% &   5031 & 3727 & 106.17s & 289.32s &  3679137\\
			\hline
			CE & 8.156s &   0.0 & 387.92s &   0.0 &   4977 & - & - & - &  7998000\\
			SF5-3/5 & 0.675s & 91.7\% & 207.46s & 46.5\% &   4977 & 3125 & 59.41s & 242.44s &  3868544\\
			SF5-4/5 & 0.671s & 91.8\% & {\bf 203.30s} & {\bf 47.6\%} &   4977 & 3421 & 79.83s & 265.60s &  3712166\\
			SF5-5/5 & 0.643s & 92.1\% & 205.79s & 46.9\% &   4977 & 3719 & 106.31s & 288.86s &  3624778\\
			\hline
			CE & 7.927s &   0.0 & 395.69s &   0.0 &   5062 & - & - & - &  7998000\\
			SF5-3/5 & 0.662s & 91.6\% & 213.31s & 46.1\% &   5062 & 3058 & 54.95s & 238.39s &  3988324\\
			SF5-4/5 & 0.731s & 90.8\% & {\bf 210.17s} & {\bf 46.9\%} &   5062 & 3385 & 77.36s & 263.75s &  3813221\\
			SF5-5/5 & 0.684s & 91.4\% & 213.29s & 46.1\% &   5062 & 3766 & 110.47s & 293.89s &  3696991\\
			\hline
			CE & 7.802s &   0.0 & 391.04s &   0.0 &   5028 & - & - & - &  7998000\\
			SF5-3/5 & 0.632s & 91.9\% & 210.46s & 46.2\% &   5028 & 3045 & 54.67s & 235.51s &  3963106\\
			SF5-4/5 & 0.765s & 90.2\% & {\bf 206.85s} & {\bf 47.1\%} &   5028 & 3442 & 82.39s & 266.24s &  3746261\\
			SF5-5/5 & 0.617s & 92.1\% & 208.99s & 46.6\% &   5028 & 3772 & 109.00s & 291.92s &  3654401\\
			\hline
\end{tabular}
\end{center}	
\caption{Runs to determine the best $\beta$. Above: ${\cal A}_g$. Below: stronger ${\cal A}_g$. Five instances, $n=4000$.}
\label{tab:beta4000}
\end{table}

\begin{table}
\begin{center}
\begin{tabular}{|lccccccccr|}
			\hline
			algo & $t_{100}$ & $\textrm{impr}_{100}$ & $T_{tot}$ & $\textrm{impr}_{tot}$ & $L$ & $s$ & $t(\beta)$ & $t_s$ & avgMpI\\
			\hline
			CE & 20.334s &   0.0 & 1564.66s &   0.0 &   7631 & - & - & - & 17997000\\
			$\beta=3/10$ & 1.576s & 92.2\% & 813.16s & 48.0\% &   7631 & 4703 & 210.29s & 965.96s &  8764537\\
			$\beta=4/10$ & 1.623s & 92.0\% & {\bf 776.44s} & {\bf 50.4\%} &   7631 & 5205 & 279.39s & 1068.07s &  8367592\\
			$\beta=5/10$ & 1.427s & 93.0\% & 780.06s & 50.1\% &   7631 & 5690 & 379.84s & 1166.74s &  8186185\\
			\hline
			CE & 20.485s &   0.0 & 1561.55s &   0.0 &   7615 & - & - & - & 17997000\\
			$\beta=3/10$ & 1.502s & 92.7\% & 831.81s & 46.7\% &   7615 & 4584 & 196.41s & 940.10s &  8897723\\
			$\beta=4/10$ & 1.490s & 92.7\% & 805.01s & 48.4\% &   7615 & 5172 & 290.37s & 1060.82s &  8418623\\
			$\beta=5/10$ & 1.363s & 93.3\% & {\bf 775.14s} & {\bf 50.4\%} &   7615 & 5707 & 382.87s & 1170.40s &  8208940\\
			\hline
			CE & 20.643s &   0.0 & 1602.77s &   0.0 &   7809 & - & - & - & 17997000\\
			$\beta=3/10$ & 1.274s & 93.8\% & 858.89s & 46.4\% &   7809 & 4746 & 209.29s & 972.95s &  8961183\\
			$\beta=4/10$ & 1.471s & 92.9\% & 833.04s & 48.0\% &   7809 & 5242 & 304.18s & 1074.07s &  8591764\\
			$\beta=5/10$ & 1.304s & 93.7\% & {\bf 817.93s} & {\bf 49.0\%} &   7809 & 5688 & 380.26s & 1167.16s &  8426212\\
			\hline
			CE & 20.376s &   0.0 & 1586.51s &   0.0 &   7740 & - & - & - & 17997000\\
			$\beta=3/10$ & 1.302s & 93.6\% & 821.90s & 48.2\% &   7740 & 4664 & 190.58s & 955.46s &  8956251\\
			$\beta=4/10$ & 1.296s & 93.6\% & {\bf 796.41s} & {\bf 49.8\%} &   7740 & 5320 & 298.60s & 1089.72s &  8443221\\
			$\beta=5/10$ & 1.293s & 93.7\% & 803.48s & 49.4\% &   7740 & 5765 & 390.54s & 1180.86s &  8283301\\
			\hline
			CE & 20.220s &   0.0 & 1572.16s &   0.0 &   7679 & - & - & - & 17997000\\
			$\beta=3/10$ & 1.310s & 93.5\% & 807.43s & 48.6\% &   7679 & 4692 & 195.04s & 960.19s &  8861558\\
			$\beta=4/10$ & 1.314s & 93.5\% & {\bf 784.95s} & {\bf 50.1\%} &   7679 & 5228 & 282.75s & 1069.80s &  8441652\\
			$\beta=5/10$ & 1.390s & 93.1\% & 785.08s & 50.1\% &   7679 & 5714 & 383.28s & 1170.17s &  8250685\\
			\hline
			\hline
			\hline
			CE & 21.537s &   0.0 & 1568.48s &   0.0 &   7631 & - & - & - & 17997000\\
			SF5-3/5 & 1.536s & 92.9\% & 807.99s & 48.5\% &   7631 & 4743 & 214.28s & 976.33s &  8690485\\
			SF5-4/5 & 1.441s & 93.3\% & 767.28s & 51.1\% &   7631 & 5245 & 275.83s & 1080.12s &  8289026\\
			SF5-5/5 & 1.445s & 93.3\% & {\bf 766.23s} & {\bf 51.1\%} &   7631 & 5759 & 382.25s & 1185.20s &  8080917\\
			\hline
			CE & 20.306s &   0.0 & 1564.66s &   0.0 &   7615 & - & - & - & 17997000\\
			SF5-3/5 & 1.315s & 93.5\% & 801.33s & 48.8\% &   7615 & 4633 & 185.37s & 949.92s &  8816084\\
			SF5-4/5 & 1.364s & 93.3\% & 773.51s & 50.6\% &   7615 & 5214 & 278.38s & 1068.87s &  8339003\\
			SF5-5/5 & 1.426s & 93.0\% & {\bf 773.02s} & {\bf 50.6\%} &   7615 & 5762 & 390.19s & 1182.55s &  8106742\\
			\hline
			CE & 20.198s &   0.0 & 1599.82s &   0.0 &   7809 & - & - & - & 17997000\\
			SF5-3/5 & 1.286s & 93.6\% & 828.43s & 48.2\% &   7809 & 4769 & 203.89s & 975.64s &  8906264\\
			SF5-4/5 & 1.262s & 93.7\% & {\bf 804.81s} & {\bf 49.7\%} &   7809 & 5279 & 285.02s & 1080.09s &  8517504\\
			SF5-5/5 & 1.261s & 93.8\% & 807.95s & 49.5\% &   7809 & 5739 & 381.82s & 1174.57s &  8338816\\
			\hline
			CE & 21.275s &   0.0 & 1588.44s &   0.0 &   7740 & - & - & - & 17997000\\
			SF5-3/5 & 1.358s & 93.6\% & 837.70s & 47.3\% &   7740 & 4767 & 213.47s & 978.97s &  8819357\\
			SF5-4/5 & 2.897s & 86.4\% & 793.45s & 50.0\% &   7740 & 5337 & 300.70s & 1095.57s &  8371821\\
			SF5-5/5 & 1.261s & 94.1\% & {\bf 786.02s} & {\bf 50.5\%} &   7740 & 5837 & 396.37s & 1197.37s &  8175083\\
			\hline
			CE & 20.884s &   0.0 & 1577.89s &   0.0 &   7679 & - & - & - & 17997000\\
			SF5-3/5 & 1.333s & 93.6\% & 809.86s & 48.7\% &   7679 & 4749 & 203.14s & 975.79s &  8771193\\
			SF5-4/5 & 1.323s & 93.7\% & {\bf 778.68s} & {\bf 50.7\%} &   7679 & 5303 & 291.59s & 1089.19s &  8342284\\
			SF5-5/5 & 1.275s & 93.9\% & 784.07s & 50.3\% &   7679 & 5814 & 402.00s & 1193.69s &  8136053\\
			\hline
		\end{tabular}
\end{center}
\caption{Runs to determine the best $\beta$. Above: ${\cal A}_g$. Below: stronger ${\cal A}_g$. Five instances, $n=6000$.}
\label{tab:beta6000}
\end{table}

\begin{table}
\begin{center}
\begin{tabular}{|lccccccccr|}
			\hline
			algo & $t_{100}$ & $\textrm{impr}_{100}$ & $T_{tot}$ & $\textrm{impr}_{tot}$ & $L$ & $s$ & $t(\beta)$ & $t_s$ & avgMpI\\
			\hline
			CE & 39.623s &   0.0 & 4144.64s &   0.0 &  10310 & - & - & - & 31996000\\
			$\beta=3/10$ & 2.794s & 92.9\% & 2156.99s & 48.0\% &  10310 & 6338 & 533.90s & 2549.07s & 15535904\\
			$\beta=4/10$ & 2.532s & 93.6\% & 2080.01s & 49.8\% &  10310 & 7060 & 752.10s & 2839.14s & 14778140\\
			$\beta=5/10$ & 2.503s & 93.7\% & {\bf 2039.32s} & {\bf 50.8\% }&  10310 & 7720 & 997.89s & 3104.02s & 14428083\\
			\hline	
			CE & 40.107s &   0.0 & 4148.15s &   0.0 &  10273 & - & - & - & 31996000\\
			$\beta=3/10$ & 2.469s & 93.8\% & 2077.86s & 49.9\% &  10273 & 6399 & 515.50s & 2583.15s & 15321545\\
			$\beta=4/10$ & 2.442s & 93.9\% & {\bf 2030.09s} & {\bf 51.1\%} &  10273 & 7041 & 725.01s & 2842.31s & 14675128\\
			$\beta=5/10$ & 2.436s & 93.9\% & 2035.16s & 50.9\% &  10273 & 7621 & 962.40s & 3075.25s & 14387203\\
			\hline
			CE & 40.706s &   0.0 & 4186.00s &   0.0 &  10325 & - & - & - & 31996000\\
			$\beta=3/10$ & 3.009s & 92.6\% & 2156.82s & 48.5\% &  10325 & 6284 & 509.45s & 2545.75s & 15620156\\
			$\beta=4/10$ & 2.691s & 93.4\% & {\bf 2070.03s} & {\bf 50.5\%} &  10325 & 7073 & 751.12s & 2864.03s & 14835548\\
			$\beta=5/10$ & 2.669s & 93.4\% & 2073.94s & 50.5\% &  10325 & 7684 & 997.36s & 3111.89s & 14511706\\
			\hline
			CE & 40.999s &   0.0 & 4389.51s &   0.0 &  10476 & - & - & - & 31996000\\
			$\beta=3/10$ & 3.009s & 92.7\% & 2198.15s & 49.9\% &  10476 & 6396 & 543.12s & 2597.42s & 15711939\\
			$\beta=4/10$ & 2.691s & 93.4\% & {\bf 2119.14s} & {\bf 51.7\%} &  10476 & 7178 & 779.27s & 2914.66s & 14924554\\
			$\beta=5/10$ & 2.727s & 93.3\% & 2128.44s & 51.5\% &  10476 & 7769 & 1029.95s & 3155.46s & 14636145\\
			\hline
			CE & 39.897s &   0.0 & 4198.46s &   0.0 &  10452 & - & - & - & 31996000\\
			$\beta=3/10$ & 2.377s & 94.0\% & 2183.58s & 48.0\% &  10452 & 6301 & 504.86s & 2528.06s & 15801765\\
			$\beta=4/10$ & 2.663s & 93.3\% & 2144.16s & 48.9\% &  10452 & 7027 & 761.81s & 2820.49s & 15056845\\
			$\beta=5/10$ & 2.503s & 93.7\% & {\bf 2122.64s} & {\bf 49.4\%} &  10452 & 7701 & 1015.08s & 3091.68s & 14743092\\
			\hline
			\hline
			\hline
			CE & 40.621s &   0.0 & 4211.89s &   0.0 &  10310 & - & - & - & 31996000\\
			SF5-3/5 & 2.892s & 92.9\% & 2183.92s & 48.1\% &  10310 & 6394 & 557.21s & 2608.32s & 15407270\\
			SF5-4/5 & 3.283s & 91.9\% & 2079.49s & 50.6\% &  10310 & 7117 & 775.28s & 2904.39s & 14645139\\
			SF5-5/5 & 2.591s & 93.6\% & {\bf 2054.86s} & {\bf 51.2\%} &  10310 & 7827 & 1043.87s & 3194.56s & 14268398\\
			\hline
			CE & 40.189s &   0.0 & 4201.60s &   0.0 &  10273 & - & - & - & 31996000\\
			SF5-3/5 & 2.468s & 93.9\% & 2079.91s & 50.5\% &  10273 & 6434 & 513.44s & 2638.89s & 15221943\\
			SF5-4/5 & 2.395s & 94.0\% & 2026.34s & 51.8\% &  10273 & 7121 & 743.18s & 2920.47s & 14528187\\
			SF5-5/5 & 2.444s & 93.9\% & {\bf 2024.17s} & {\bf 51.8\%} &  10273 & 7720 & 982.99s & 3163.26s & 14224719\\
			\hline
			CE & 40.854s &   0.0 & 4256.59s &   0.0 &  10325 & - & - & - & 31996000\\
			SF5-3/5 & 4.029s & 90.1\% & 2165.45s & 49.1\% &  10325 & 6329 & 514.73s & 2605.41s & 15510828\\
			SF5-4/5 & 2.856s & 93.0\% & {\bf 2100.96s} & {\bf 50.6\%} &  10325 & 7151 & 796.75s & 2943.48s & 14687474\\
			SF5-5/5 & 2.739s & 93.3\% & 2107.66s & 50.5\% &  10325 & 7792 & 1066.36s & 3209.07s & 14335403\\
			\hline
			CE & 41.832s &   0.0 & 4276.16s &   0.0 &  10476 & - & - & - & 31996000\\
			SF5-3/5 & 2.696s & 93.6\% & 2200.29s & 48.5\% &  10476 & 6442 & 544.56s & 2629.27s & 15602750\\
			SF5-4/5 & 2.992s & 92.8\% & 2167.29s & 49.3\% &  10476 & 7241 & 837.21s & 2956.36s & 14783414\\
			SF5-5/5 & 4.052s & 90.3\% & {\bf 2123.13s} & {\bf 50.3\%} &  10476 & 7877 & 1062.49s & 3217.01s & 14457500\\
			\hline
			CE & 41.828s &   0.0 & 4269.96s &   0.0 &  10452 & - & - & - & 31996000\\
			SF5-3/5 & 2.426s & 94.2\% & 2192.43s & 48.7\% &  10452 & 6370 & 519.82s & 2598.16s & 15665979\\
			SF5-4/5 & 3.408s & 91.9\% & {\bf 2134.77s} & {\bf 50.0\%} &  10452 & 7099 & 762.58s & 2897.57s & 14920109\\
			SF5-5/5 & 2.397s & 94.3\% & 2139.26s & 49.9\% &  10452 & 7812 & 1059.19s & 3188.45s & 14578944\\
			\hline
	\end{tabular}
\end{center}
\caption{Runs to determine the best $\beta$. Above: ${\cal A}_g$. Below: stronger ${\cal A}_g$. Five instances, $n=8000$.}
\label{tab:beta8000}
\end{table}

\paragraph{Minor improvements and technical details.}
For the sake of simplicity, in this paper we have presented the algorithm in a basic form although some small improvements 
are still possible. In particular, we can strengthen the definition of a \good\ edge in such a way that it becomes more
difficult for an edge to be expanded. Furthermore, in the basic form discussed so far, some moves $\mu(i,j)$ could be evaluated twice, and this happens when 
both edges $\{i,i+1\}$ and $\{j,j+1\}$ are \good. There is a simple workaround to avoid this double evaluation, and it is described in Appendix B, together with the stricter definition of \good\ edges.
We remark that the theoretical results, in particular the good average-case complexity, are already obtained by the basic form. Furthermore, we ran some computational experiments to evaluate the possible improvements and found out that their impact is quite minor (say, less than $10\%$ of time savings) while the algorithm becomes definitely less readable. For this reason, we have chosen to briefly describe these minor improvements only in the Appendix.

\section{Conclusions}
\label{sec:concl}

In this work we have described a new exact strategy for finding the
best \opt{2} move in a given tour. A nice feature of this strategy is, generally speaking, its simplicity. In particular, the blind version of our algorithm amounts to a 1-line code change to the standard \opt{2} procedure which is  sufficient to lower the complexity of of this step for many local search iterations. Improving the \opt{2} search (a popular algorithm) for the TSP (a fundamental problem) can have 
a certain impact, especially since the this \myquote{hack} is so trivial that can be readily incorporated   in many simple applications previously based on the basic 2-OPT local search.
We have also described a family of sub-quadratic average-case heuristics for the same problem, which can be tuned in such a way that they succeed with very high probability. 
Computational  experiments and theoretical analysis have shown that our strategies
outperform the classical two-nested-for algorithm 
for a good part of a local search convergence starting from a random tour. In particular, on a starting random uniform instance, we determine the best move  in average time $O(n^\frac{3}{2})$, while on a Euclidean instance, our procedure takes average linear time, which is the best possible 
complexity for this problem. 

We have then discussed how to adjust our procedure to obtain an effective local search algorithm, given that the performance worsens while we approach the local optima. We have therefore proposed a hybrid approach, made by our procedure for the first part followed by the standard algorithm for the rest of the convergence, and we have studied the best point at which we should make the switch. 

A direction for future research would be to integrate our algorithm, which is extremely fast in the first local search iterations, with
some ad-hoc improvements for the second part, in order to further improve its overall performance.


%
%
%

\bibliographystyle{abbrv}
\bibliography{tspbib}

	\section*{APPENDIX A: On the distance between random points}
	\label{app:lemma}
	
	\begin{lemma}
		\label{lem:maxd}
		Let $1.055 < d \le \sqrt{2}$, let $X$, $X'$ be two random points in the unit square and $D=||X-X'||$. Then  
		\[
		\prob{D > d}\le  \frac{7}{16} \left(1 - \sqrt{d^2 - 1}\right)^4.
		\]
	\end{lemma}
	\myproof{
		Consider Figure \ref{fig:probD}.
		In order for two points to have distance greater than $d$, they must not fall within a circle of ray $d/2$. We draw such a circle with center in $(0,0)$ and look at the
		intersections of the circle and the unit square.
		From Pythagoras' theorem, we get 
		\[
		y = \frac{1}{2}\sqrt{d^2-1}
		\]
		and then
		\[
		z = \frac{1}{2}-y = \frac{1}{2}\left(1- \sqrt{d^2-1}\right).
		\]
		For two points to have distance greater than $d$ at least one of them should fall out of the
		circle, i.e., in the corners, each of which is an area of \myquote{triangular} shape but with a curve basis. We are going to relax this, and require that the point must fall within one of the four triangles with two sides of length $z$ in the corners (this way we are overestimating the probability).
		Let us call $T^1,\ldots,T^4$ these squares, starting from the top-left and proceeding counter-clockwise.
		
		Once a point is in $T\in\{T^1,\ldots,T^4\}$, the other point must be at distance greater than $d$ from it.  
		The distance between two points in triangles which are not opposite to each other is at most $d$,
		as it can be checked
		by noticing that $\sqrt{1+z^2} \le d$ is always satisfied for $d\ge 1.055$ (where $\sqrt{1+z^2}$ is the maximum distance between two
		non-opposite triangles $T$). 
		Therefore we have to look for the second point in the opposite corner.

		\begin{figure}[t!]
			\caption{Study of $\prob{D>d}$.}
			\label{fig:probD}
			
			{\begin{center}
					\begin{tikzpicture}[scale=6]
						\draw [dashed](0,0) circle (0.6cm);
						\draw (-0.5,-0.5) -- (1/2,-1/2) -- (1/2,1/2) -- (-1/2,1/2) -- (-1/2,-1/2);
						\draw (0,0)--(0.5, 0.331);
						\draw (0,0)--(0.331,0.5);
						\draw [dotted] (0.331,0.5)--(0.331,0.331);
						\draw [dotted] (0.331,-0.331)--(0.331,-0.5);
						\draw [dotted] (-0.331,0.5)--(-0.331,0.331);
						\draw [dotted] (-0.331,-0.331)--(-0.331,-0.5);
						\draw [dotted] (-0.5,0.331)--(-0.331, 0.331);
						\draw [dotted] (0.331,0.331)--(0.5, 0.331);
						\draw [dotted] (-0.5,-0.331)--(-0.331, -0.331);
						\draw [dotted] (0.331,-0.331)--(0.5, -0.331);
						\draw (0.331,0.5)--(0.5,0.331);
						\draw (-0.331,0.5)--(-0.5,0.331);
						\draw (-0.331,-0.5)--(-0.5,-0.331);
						\draw (0.331,-0.5)--(0.5,-0.331);
						\draw (-1/2+0.338,1/2)--(-1/2,1/2-0.338);
						
						\draw (1/2,-1/2) --(-1/2+0.338,1/2);
						\draw (1/2,-1/2) --(-1/2,1/2-0.338);
						\draw [->] (-3/4,0)--(3/4,0);
						\draw [->] (0,-3/4)--(0,3/4);
						\node at (0.525,0.16) {$y$};
						\node at (0.525,0.42) {$z$};
						\node at (0.16,0.525) {$y$};
						\node at (0.42,0.525) {$z$};
						\node at (0.25,0.22) {$\frac{d}{2}$};
						\node at (-0.03,-0.04) {0};
						\node at (-0.03,0.55) {$\frac{1}{2}$};
						\node at (-0.035,-0.55) {-$\frac{1}{2}$};
						\node at (0.46,-0.05) {$\frac{1}{2}$};
						\node at (-0.45,-0.05) {-$\frac{1}{2}$};
						\node at (-0.15,-0.16) {$d$};
						\node at (0.53,-0.53) {$V$};
						\node at (-0.53,0.53) {$V'$};
						
						\tikzset{shift={(1/2,-1/2)}}
						\draw[dashed] (123.5:1.2) arc (123.5:146.5:1.2) ;
					\end{tikzpicture}
					
				\end{center}
			}
		\end{figure}
		
		Since in $T$ the farthest from other points is precisely the corner vertex, say $V$, we draw
		a circle $C_V$ of ray $d$ and center $V$. Let $V'$ be
		the opposite corner, and  
		let $l$ be the distance between $V'$ and the intersection of the circle $C_V$ with the square. 
		By Pythagoras' theorem we get
		\[
		l = 1 - \sqrt{d^2 - 1} = 2z.
		\] 
		Let $L^1,\ldots,L^4$ be the triangle with two sides of length $l$ in the corners.
		From the previous discussion, if we pick a sequence $X$, $X'$ of two points at random, in order to have
		$D>d$ there must exist opposite corners $i$ and $j$ such that either $X$ in $T^i$ and $X'$ is in $L^j$, or $X$ is in $L^i\setminus T^i$ and $X'$ is
		in $T^j$ (notice that since we used 
		chords instead of circle arcs in these areas, these are necessary, but not 
		sufficient conditions, and so we are overestimating
		the probability that $D>d$).
		For fixed opposite $i,j$, we have
		\[
		\prob{(X\in T^i \land X'\in L^j) \lor (X\in L^i\setminus T^i \land X'\in T^j)} 
		= \frac{z^2}{2}\cdot \frac{4z^2}{2} + \frac{3z^2}{2}\cdot \frac{z^2}{2} = \frac{7}{4}z^4.
		\]
		Therefore, for the probability that the above situation is realized at one of the 4 possible pairs of opposite corners, we obtain
		\[
		\prob{D > d}\le 7 z^4  = \frac{7}{16} \left(1 - \sqrt{d^2 - 1}\right)^4.
		\]
	}
	
\section*{APPENDIX B: Minor improvements to the procedure}
	\label{app:tech}
	
	\newcommand{\cmin}[1]{c_{\min}(#1)}
	\paragraph{A slightly stronger condition for expansion.}
	
	In our research, we also investigated a stronger threshold to define a pivot as ``\good''. 
	For each vertex $u$, let us then denote by 
	\[
	\cmin{u}:= \min_v c(u,v)
	\]
	the minimum weight of any edge incident in $u$. Notice that for each edge $\{u,v\}$ it is
	\[
	c(u,v) \ge \frac{\cmin{u}+\cmin{v}}{2}
	\]

	Given the tour $(1,2,\ldots,n)$, let us define for all $i$
	\[
	c'(i,i+1) := c(i,i+1) - \frac{\cmin{i} + \cmin{i+1}}{2}
	\]

	Assume the current champion is $\hat \mu=\mu(\hat \imath , \hat \jmath)$.
	Then, for any move $\mu(i,j)$ better than $\hat \mu$ it must be
	\begin{align*}
		\Delta(\hat \mu)  & < \Delta(\mu(i,j)) \\
		& =  c(i,i+1) + c(j,j+1) - 
		\big(c(i,j) + c(i+1,j+1)\big) \\
		& \le  c(i,i+1) + c(j,j+1) - \frac{\cmin{i} + \cmin{j} + \cmin{i+1} + \cmin{j+1}}{2}\\
		& = c'(i,i+1) + c'(j,j+1)
	\end{align*}
	and hence
	\[
	\left (c'(i,i+1) > \frac{\D(\hat \mu)}{2}\right) \lor \left(c'(j,j+1) > \frac{\D(\hat \mu)}{2}\right)  
	\]
	i.e., the move must be 
	$(\D(\hat \mu)/2)$-\good\ with respect to the values $c'(v,v+1)$, rather than to the values $c(v,v+1)$ that were used before, and,  clearly, $c'(v,v+1) < c(v,v+1)$.
	
	As far as the complexity of this strategy is concerned, notice that computing the
	values $\cmin{\cdot}$ takes $\Theta(n^2)$ time, but can be
	amortized since we need do it only once, at the beginning of a local search
	convergence. Once the $\cmin{\cdot}$ values have been computed in
	this preprocessing step, computing $c'(v,v+1)$ is $O(1)$.
	
	\paragraph{Avoiding double evaluation of moves.} If no precaution is taken, any move $\mu(i,j)$ could be possibly evaluated twice, namely, when both $\{i,i+1\}$ and $\{j,j+1\}$ are \good. In order to avoid this, we can keep a partially-filled array $\texttt{neverExp[]}$ and a counter $\texttt{c}$ giving its current size. The idea is that $\texttt{neverExp[1..c]}$ contains all pivots that were not expanded so far. Initially, $\texttt{c}=n$ and  $\texttt{neverExp[]}$ contains all indices $1,\ldots,n$. 
	In general, assume an edge $\{i,i+1\}$ is being expanded. Then, to create the moves $\mu(i,j)$, 
	instead of considering  $j\leftarrow\texttt{k}$ for all $\texttt{k}=1,\ldots,n$ 
	we consider $j\leftarrow\texttt{neverExp[k]}$ for $\texttt{k}=1,\ldots,\texttt{c}$. 
	When $j\ne i$, then we create the move $\mu(i,j)$ and increase $\texttt{k}$,
	but when $j=i$, we simply overwrite $\texttt{neverExp[k]}$ with
	$\texttt{neverExp[c]}$, decrease $\texttt{c}$ and stay at the same $\texttt{k}$. This way the array always contains indices that were never expanded, and so no moves can be  evaluated twice.  Furthermore, since the order of the pivots is irrelevant, we were able to make this update in time $O(1)$ and so the complexity of an expansion remains $\Theta(n)$ for all expansions.

\end{document}